\documentclass[10pt]{article}
\usepackage{graphicx}
\usepackage{dcolumn}
\usepackage{bm,color}
\usepackage{amsmath, latexsym, amssymb, amscd, amsfonts, amsthm, epsfig, mathrsfs, graphicx, url,verbatim}
\usepackage[english]{babel}
\usepackage[braket]{qcircuit}

\DeclareFixedFont{\sfracFont}{U}{euf}{b}{n}{7pt}
\newtheorem{theorem}{Theorem}
\newtheorem{lemma}{Lemma}
\newtheorem{definition}{Definition}
\newtheorem{assumption}{Assumption}
\newtheorem{corollary}{Corollary}
\newtheorem{remark}{Remark}

\newtheoremstyle{mytheo}
  {16pt}
  {0pt}
  {}
  {}
  {\bfseries}
  {:}
  {.5em}
  {}

\theoremstyle{mytheo}
\newtheorem{example}{Example}

\begin{document}
\title{Discretisation of quantum feedback networks for implementation on a quantum computer}

\author{Luc Bouten and John Gough}


\date{\today}
\maketitle

\begin{abstract}
Quantum Feedback Network Theory (also known colloquially as the $SLH$-framework) is a powerful tool for modeling  systems in quantum optics. This paper describes a method for discretising $SLH$-networks such that they can be implemented on a quantum computer. Building on \cite[Thm 1]{BvH08}, we establish strong convergence, uniformly on compact time intervals, of the discretised system to the continuous $SLH$-network as the 
discretisation mesh goes to zero.
\end{abstract}

\maketitle

\section{Introduction}

The starting point of our investigation is  
a Hudson-Parthasarathy quantum stochastic differential equation (QSDE) 
\cite{HuP84}. It describes a unitary 
process that can be interpreted as a Markovian approximation of the interaction 
of a system, e.g.\ an atom, a collection of atoms or optical elements such as beam 
splitters or mirrors, with the electromagnetic field \cite{AFLu90, Gou05}. 
As we will see, the HP-equation is characterized by its triple of coefficients $(S,L,H)$ 
where $S$ is a scattering-matrix of system operators describing the effect on the system of scattering photons between different field-channels, $L$ is a vector of coupling-operators of the system to the field channels, and  $H$ is the Hamiltonian describing the internal evolution of the system. 

In the Hudson-Parthasarathy theory, the electromagnetic field is modeled as a collection of channels supporting photons that move through the channel at the speed of light. The photons interact with the system at the origin. In this way, we can think of a field channel as an \emph{input}, the photons that are flying towards the system at the origin, and an \emph{output}, the photons that have already interacted with the system at the origin and are flying away. In quantum optics, this picture has been introduced by Gardiner and Collett \cite{GC84} and is known under the name input-output formalism. Photons move at the speed of light and, for convenience, we set $c=1$; we may therefore identify the distance of a photon to the system at the origin with the time it takes to get to the origin, interpreting negative times as the time that has elapsed since the photon passed the origin.  
 
Let us now look at a situation where we have $N$ systems at the origin (i.e.~at $t=0$) each connected to $m_i$ field channels ($i =1,\ldots,N$). That is, the $i$th system has $m_i$ inputs and $m_i$ outputs. A natural step is then to connect some of the outputs to some of the inputs. In this way, creating internal lines in which photons can fly from one 
system to another, itself included.  This situation was first introduced by Yanagisawa and Kimura for a specific example system in \cite{YaK03a,YaK03b} and studied at the level of transfer functions. Closely related is also 
the work by Gardiner \cite{Gar93} and Carmichael \cite{Car93a} on cascaded systems.

A first attempt at a rigorous general mathematical definition of the situation described above was provided by Gough and James  \cite{GoJ09b, GoJ09}. They define a quantum feedback network by simply taking the sum of the Chebotarev-Gregoratti Hamiltonians \cite{Che97, Gre01} of the subcomponents, creating inner lines by connecting inputs to outputs, and gathering all boundary conditions of the sub-components to define the domain of the network Hamiltonian. Unfortunately, they leave the essential self-adjointness of the network Hamiltonian on this domain open as a conjecture \cite{GoJ09b}.    
Due to the standard notation $S$, $L$ and $H$ for the coefficients of a Hudson-Parthasarathy QSDE, 
the quantum feedback networks of Gough and James came to be known as $SLH$-networks, see also 
the review \cite{CKS17} and the references therein.

To circumvent the problems with the essential self-adjointness of the network Hamiltonian on its domain, we resort to an alternate definition of the time evolution of the network. We simply work on the level of the unitary co-cycle originating from the QSDE, which we periodically interrupt to transfer the output of channels that are connected to an input channel, to the respective input channel. This is done by defining a so-called braid map $B_\sigma$ for a given time delay $\sigma \ge0$. We use the co-cycle property to split the evolution in a part before and a part after the interruption. The interruption itself is represented by the unitary braid map $B_\sigma $.  
This rigorous definition of an $SLH$-network is a new contribution by this paper.

There is extensive literature on discrete approximation of QSDEs by two-level systems: 
weak convergence was covered by Parthasarathy \cite{Par88}, a significant step forward was 
provided by Lindsay and Parthasarathy \cite{LiP88} showing weak convergence of the quantum flow, 
and Attal and Pautrat \cite{AtP06} proved strong convergence uniformly on compact time 
intervals for very specific discrete noise models.  Attal and Pautrat \cite{AtP06} cover the multiple channel case 
by using $m+1$ dimensional systems: one vector representing the vacuum and the other $m$ vectors 
each representing an excitation in the $m$ respective channels.    
In quantum computer simulations of QSDEs \cite{ViB19, BVS19}, however, 
each channel is usually represented by its own copy of a qubit in each time step. 
That is, each channel is represented by a string of qubits. 
In this situation, in each time step, the noise is given by a $2^m$ dimensional system, 
which is outside the Attal and Pautrat framework for $m>1$. 
In \cite{BvH08}, a theory is developed for approximating QSDEs using 
strong convergence uniformly on compact time intervals for a very rich class of noise models and 
QSDE coefficients (even unbounded coupling operators are allowed).

The main result in this paper, Theorem \ref{theorem main result}, 
is an extension of the convergence theorem \cite[Thm 1]{BvH08} to the setting of $SLH$-networks. 
Theorem \ref{theorem main result} can be used to show strong convergence uniformly on compact 
time intervals of the time evolution of a digital twin of an $SLH$-network, running on a quantum computer, 
to the dynamics of the real physical $SLH$-network as the discretisation mesh goes to zero. In practice though, 
the easiest way to use Theorem \ref{theorem main result} is through its Corollary \ref{corollary}.  
Corollary \ref{corollary} asserts that proving convergence for each of the components of an $SLH$-network separately, using 
the convergence theorem \cite[Thm 1]{BvH08}, automatically leads to convergence of the whole $SLH$-network 
(provided the technical condition Assumption \ref{assumption multiples of tau} is in force). Loosely speaking, it states that 
the composition of a network from its components commutes with its discretisation.

The convergence results presented in this paper open 
the door to quantum hardware simulations 
for a large class of quantum optical systems. All systems that can 
be represented as an $SLH$-network, including both cavity and circuit QED, 
can in principle be discretised and implemented on a quantum computer.
The results in this paper guarantee that as the discretisation mesh gets 
finer and finer, the discrete approximation running on the quantum computer 
resembles the real continuous physics better and better. 
Of course, this assumes that at some point in the future sufficiently 
large and sufficiently reliable quantum computers will become available.
Note, however, that $SLH$-networks that can be simulated using only a limited 
amount of qubits can already be studied on classical simulators 
of quantum computers, see e.g.\ \cite{BVS19}, 
rendering the results in this paper immediately relevant. 

The remainder of this paper is organized as follows. 
Section \ref{section SLH triple} introduces the main objects: 
it defines an $SLH$-triple and introduces the unitary co-cycles 
that govern the time evolution of the  network 
components as solutions to the QSDE for the respective network node. 
Section \ref{section SLH-networks} defines a network as a collection of components with connections 
between certain outputs and inputs. Furthermore, it introduces the braid map that transfers the 
output of channels that are connected to an input channel, to the respective 
input channel. Subsequently the section introduces the time evolution of the 
connected network. Section \ref{section convergence theorem} contains the 
definition of the generator associated with an $SLH$-network and it defines 
the discretised time evolution, represented on the space of the continuous 
time network. It then turns to the main result: Theorem \ref{theorem main result} 
which demonstrates the convergence of the discrete time evolution to 
the continuous time evolution of the network as the discretisation 
mesh goes to zero. Section \ref{section intermezzo} is a short intermezzo showing 
how to extend to the case of unbounded coefficients. Section \ref{section SLH on the quantum computer} shows how 
to use the main result, Thm  \ref{theorem main result}, in the setting of a simulation on a quantum computer.
Sections \ref{section examples components} and \ref{section examples networks} 
provide examples, showing how to implement a simple cavity QED simulation. 
The last section, Section \ref{section proofs}, contains the proof of the main result.

\section{The quantum stochastic calculus}\label{section SLH triple}

We start our investigation of $SLH$-networks by first studying a single component of such a network with inputs and outputs that are still unconnected.
Such a single component of an $SLH$-network consists of a quantum system represented on some separable complex Hilbert space $\mathcal{H}$ and $m$ input-output channels  that all interact with each other at the origin. We call $m$ the multiplicity of the component.

For an interval $I\subset \mathbb{R}$, we define the \emph{symmetric Fock space with multiplicity} $m$ as
\begin{equation}\label{def Fock space}
\mathcal{F}_I = \mathbb{C} \oplus \bigoplus_{n =1}^\infty L^2(I; \mathbb{C}^m)^{\otimes_s n}.
\end{equation}
Here, $L^2(I; \mathbb{C}^m)$ stands for the Hilbert space of  all quadratically integrable 
functions on $I$ with values in $\mathbb{C}^m$. We have used the symmetric tensor product which means that the $n^{th}$ layer in the direct sums describes $n$ bosonic particles, i.e.\ $n$ particles with symmetrized wave functions. 
Note that the Fock space $\mathcal{F}_I$ describes a field of bosons, which we will always take to be a field of photons. The $m$ input-output channels are represented on the symmetric Fock space $\mathcal{F}_\mathbb{R}$ and should be interpreted as $m$ channels 
in the electromagnetic field. We will often shorten $\mathcal{F}_\mathbb{R}$
to $\mathcal{F}$.
The entire network-component is then represented on the Hilbert space $\mathcal{H}\otimes\mathcal{F}$.

For $f \in L^2(I; \mathbb{C}^m)$, we define the \emph{exponential vector} $e(f) \in \mathcal{F}_I$ as
\begin{equation}\label{def exponential vector}
e(f) = 1 \oplus \bigoplus_{n= 1}^\infty \frac{1}{\sqrt{n!}}f^{\otimes n}.
\end{equation}
We call the linear span of the exponential vectors the \emph{exponential domain} which is dense in $\mathcal{F}_I$.

We are now ready to introduce for all  $1\le i,j \le m$ and $t\ge 0$ the fundamental noises $A^i_t,\ A^{i*}_t$ and $\Lambda^{ij}_t$ on the exponential domain of $\mathcal{F}$
by (see also \cite{HuP84, Par92, BaL00})
  \begin{equation}\label{def fundamental noises}\begin{split}
  &A^i_t e(f) = \left(\int_0^t f_i(s)ds\right)\, e(f),\\
  &\big\langle e(g), A^{i*}_t e(f)\big\rangle  = 
  \left(\int_0^t\overline{g}_i(s)ds\right)\, \big\langle e(g),
  e(f)\big\rangle, \\
  &\big\langle e(g), \Lambda^{ij}_t e(f)\big\rangle  = 
  \left(\int_0^t\overline{g}_i(s)f_j(s)ds\right)\, 
  \big\langle e(g),e(f)\big\rangle, 
  \end{split}\end{equation}
for all $f$ and $g$ in $L^2(\mathbb{R}; \mathbb{C}^m)$. $A^i_t$ and $A^{i*}_t$ are called the \emph{annihilation} and \emph{creation} processes, respectively. The processes 
$\Lambda^{ij}_t$ are called \emph{gauge} processes.

For an interval $I \subset \mathbb{R}$, we denote by $\chi_I$ the \emph{indicator function of} $I$, i.e., the function that is 1 on I and 0 elsewhere in $\mathbb{R}$.
For all $f \in L^2(\mathbb{R}; \mathbb{C}^m)$ and
$I \subset \mathbb{R}$, we can now define $f_I$ by
\begin{equation*}
f_{I}(x) = f(x)\chi_I(x), \ \ \mbox{for all}\ \ x \in \mathbb{R}.
\end{equation*}
We will use the shorthand $f_{t)} = f_{(-\infty, t)}$ and 
$f_{[t} = f_{[t, \infty)}$ in the following. There exists 
a unique unitary isomorphism $\iota$ from 
$\mathcal{F}$ to $\mathcal{F}_{(-\infty,t)}\otimes\mathcal{F}_{[t,\infty)}$ 
such that (see e.g.\ \cite[Prop 19.6]{Par92}) 
\begin{equation*}
\iota\big(e(f)\big) = e(f_{t)})\otimes e(f_{[t}), 
\end{equation*}
for all $f \in L^2(\mathbb{R}; \mathbb{C}^m)$ and $t \ge 0$.
It is this so-called \emph{continuous tensor product structure} of the Fock space $\mathcal{F}$, that forms the starting point for the quantum stochastic calculus of Hudson and Parthasarathy \cite{HuP84, Par92}. In the following, we will simply identify $\mathcal{F}$ and $\mathcal{F}_{(-\infty,t)}\otimes\mathcal{F}_{[t,\infty)}$, i.e., to keep notation light, we will no longer write the unitary map $\iota$ between them.
A process $L_s,\, (s\ge 0)$ on $\mathcal{H}\otimes \mathcal{F}$ is called \emph{adapted} if $L_s$ acts nontrivially only on $\mathcal{H}\otimes \mathcal{F}_{(-\infty, s)}$ and 
is the identity on $\mathcal{H}\otimes \mathcal{F}_{[s, \infty)}$ for all $s\ge 0$.

It is possible \cite{HuP84}  to define stochastic integrals of adapted processes $L_s$ against the fundamental noises, i.e., to give meaning to the expression $X_t = X_0 + \int_0^tL_sdM_s$ where $M_s$ is one of the fundamental noises $A^i_s,\ A^{i*}_s$ 
or $\Lambda^{ij}_s$. The expression can be written in shorthand as $dX_t = L_tdM_t$. 
The calculus with which these stochastic integrals can be manipulated in calculations is due to Hudson and Parthasarathy \cite{HuP84}. 
The calculus consists of the following. Suppose $X_t$ and $Y_t$ are stochastic integrals, i.e., $dX_t = L^1_tdM^1_t$ and $dY_t = L^2_tdM^2_t$ where $L^1$ and $L^2$ are adapted processes and $M^1$ and $M^2$ are fundamental noises, then the product $X_tY_t$ is itself a stochastic integral. Moreover, the product $X_tY_t$ satisfies the following 
quantum Ito rule, (integration by parts rule)
  \begin{equation*}
  d(X_t Y_t) = X_tdY_t + (dX_t)Y_t + dX_tdY_t, 
  \end{equation*}
where to evaluate $dX_tdY_t$ we use that the 
increment $dM_t$ of a fundamental noise commutes 
with all adapted processes, and products $dM^1_tdM^2_t$ 
are given by the following quantum It\^o table \cite{HuP84}  
\begin{center}
{\large \begin{tabular} {l|lll}
$dM^1\backslash dM^2$ & $dA^{i*}_t$ & $d\Lambda^{ij}_t$ & $dA^i_t$ \\
\hline 
$dA^{k*}_t$ & $0$ & $0$ & $0$ \\
$d\Lambda^{kl}_t$ & $\delta_{li}dA^{k*}_t$ & $\delta_{li}d\Lambda^{kj}_t$ & $0$  \\
$dA^k_t$ & $\delta_{ki}dt$ & $\delta_{ki}dA^j_t$ & $0$ 
\end{tabular} }
\end{center}
and all products $dM_tdt$ and $dtdM_t$ are zero.
As an example, suppose $dX_t = L^1_tdA^i_t$ and 
$dY_t = L^2_tdA^{i*}_t$, then 
$d(X_tY_t) = X_tL^2_tdA^{i*}_t + L^1_tY_tdA^{i}_t + L^1_tL^2_tdt$.

We now define:
\begin{definition}
Let $m$ be an element of $\mathbb{N}\backslash \{0\}$.
Let $\mathcal{H}$ be a separable Hilbert space. 
An $SLH$-triple with multiplicity $m$ acting on $\mathcal{H}$ 
is a triple $(S,L,H)$, where  $S$ is a collection 
of bounded operators $S_{ij}, \ 1 \le i,j \le m$ on $\mathcal{H}$, 
such that
\begin{equation*}
  \sum_{j=1}^m S_{ji}^*S_{jl} = \delta_{il}, \ \ \ \ \sum_{j=1}^m S_{ij}S_{lj}^* = \delta_{il},
\end{equation*}
$L$ is a collection of bounded operators $L_i, 1\le i\le m$ on 
$\mathcal{H}$ and $H$ is a self-adjoint bounded operator 
on $\mathcal{H}$.
\end{definition}

Using the definition of the quantum stochastic integral  \cite{HuP84}, we can 
write down the following quantum stochastic 
differential equation (QSDE) on $\mathcal{H}\otimes \mathcal{F}$:
\begin{equation}\label{definition Ut}\begin{split}
    dU_t = \Bigg\{\sum_{i,j = 1}^m\big(S_{i j}-\delta_{ij}\big) d \Lambda^{ij}_t + 
      \sum_{i=1}^m L_i dA^{i*}_{t} -\sum_{i,j = 1}^m L^*_i S_{i j}dA^j_t 
    -\frac{1}{2}\sum_{i= 1}^m L^*_iL_i dt - iHdt\Bigg\}U_t, 
\end{split}\end{equation} 
with initial condition $U_0 =I$ and $t\ge 0$. Here $(S, L, H)$ is 
an $SLH$-triple with multiplicity $m$, acting on $\mathcal{H}$.
It follows from \cite{HuP84} that an adapted 
process $U_t$ that solves
equation  \eqref{definition Ut} exists and is 
unique. Furthermore, it follows from \cite{HuP84}
that the solution to equation  \eqref{definition Ut} is 
unitary if and only if the coefficients of the 
equation are as in equation  \eqref{definition Ut} 
and form an $SLH$-triple. 

We are now ready to introduce the time evolution of a single component of a network. We 
let $\theta_t$ be the \emph{left shift} on $L^2(\mathbb{R}; \mathbb{C}^m)$, i.e.\
\begin{equation*}
\theta_t(f)(x) = f(x+t)\ \ \ \  \mbox{for all}\ x \in \mathbb{R}.
\end{equation*}
We denote by $\Theta_t: \mathcal{F} \to \mathcal{F}$ the second quantisation of the map $\theta_t$. $\Theta_t$ represents the free time evolution of the $m$ input-output channels. The photons simply fly along the real axis from positive to negative in units such that $c$, the speed of light, equals $1$ which allows us to identify time $t$ and position $x$ along the channel.

The quantum system with Hilbert space $\mathcal{H}$ is 
positioned at the origin and perturbs the free 
evolution $\Theta_t$ in a way that is completely determined 
by the $SLH$-triple associated to the network-component. The $SLH$-triple 
of the network-component determines via equation  \eqref{definition Ut}
an adapted unitary process $U_t$ that is a 
co-cycle with respect to the shift $\Theta_t$, i.e.
\begin{equation}\label{cocycle Ut}
U_{t+s} = \Theta_{-s}U_t\Theta_{s} U_s.
\end{equation}
Note that $U_t$ is not a one-parameter group 
of unitaries. We can however, define a one-parameter 
group of unitaries $\hat{U}_t$ by
\begin{equation}\label{definition Uhat}
\hat{U}_t =  \left\{ \begin{array}{ll}
 \Theta_t U_t & \mbox{\ \ \ if \ } t \ge 0 \\
 U^*_{-t} \Theta_t & \mbox{\ \ \ if \ } t < 0 
  \end{array}\right. .
\end{equation}


\section{$SLH$-networks}\label{section SLH-networks}

Suppose that we have $N$ separate components with initial spaces $\mathcal{H}_1,\ldots, \mathcal{H}_N$. We next include channels with multiplicities $m_1,\ldots, m_N$, symmetric Fock spaces $\mathcal{F}_{1},\ldots, \mathcal{F}_{N}$ and $SLH$-triples $(S^1, L^1, H^1),\ldots, (S^N, L^N, H^N)$. We can now introduce:
\begin{equation*}\begin{split}
&\mathcal{H} = \mathcal{H}_1\otimes\ldots\otimes\mathcal{H}_N, \\
&m = m_1 + \ldots+m_N, \\
&\mathcal{F} = \mathcal{F}_1 \otimes\ldots\otimes\mathcal{F}_N.
\end{split}\end{equation*}
As $\mathcal{F}_i$ is the symmetric second-quantisation of $L^2(\mathbb{R}; \mathbb{C}^{m_i})$, we can again use \cite[Prop 19.6]{Par92} to exploit the canonical isometry between the tensor products of the Fock spaces and the symmetric second-quantisation of the direct sum of the initial spaces to conclude that
\begin{equation*}
\mathcal{F} = \mathbb{C}\oplus \bigoplus_{n=1}^\infty L^2(\mathbb{R}; \mathbb{C}^m)^{\otimes_s n}.
\end{equation*} 
That is, the combined system of all the components has underlying Hilbert space $\mathcal{H}\otimes\mathcal{F}$, where $\mathcal{F}$ has multiplicity $m$. By ampliation with the identity, we can now extend all operators in the $SLH$-triples of the components and all the quantum noises to the space $\mathcal{H}\otimes\mathcal{F}$. We re-label the noises such that the labels run from $1$ to $m$. That is, the noises labeled with $1,\ldots,m_1$ are the noises that act non-trivially on $\mathcal{F}_1$, the noises labeled with $m_1+1, \ldots,m_1+m_2$ are the noises that act non-trivially on $\mathcal{F}_2$ and so on up to the noises labeled by $(\sum_{i=1}^{N-1}m_i) +1,\ldots,m$ 
which are the noises that act non-trivially on $\mathcal{F}_N$.

\begin{definition}\label{definition SLH composition}
The $SLH$-triple $(S,L,H)$ of the network that consists of $N$ components with initial spaces $\mathcal{H}_1,\ldots, \mathcal{H}_N$, multiplicities $m_1,\ldots, m_N$, symmetric Fock spaces $\mathcal{F}_{1},\ldots, \mathcal{F}_{N}$ and $SLH$-triples $(S^1, L^1, H^1),\ldots, (S^N, L^N, H^N)$, is given by 
\begin{equation*}
S =
\begin{bmatrix}
S_{11} & \ldots & S_{1m} \\
\vdots & & \vdots\\
S_{m1}& \ldots & S_{mm}
\end{bmatrix} = 
\begin{bmatrix}
S^1 & 0 & \ldots & 0 \\
0 & S^2 & \ldots & 0 \\
\vdots & \vdots & & \vdots \\
0 & 0&\ldots&S^N
\end{bmatrix},\ \ \
L = \begin{bmatrix} 
L_1 \\
\vdots\\
L_m
\end{bmatrix}
= 
\begin{bmatrix}
L^1\\
\vdots \\
L^N
\end{bmatrix},\ \ \
H = \sum_{i=1}^N H^i.
\end{equation*} 
\end{definition}

The time evolution of the complete network, consisting of the $N$ components (without 
connections between inputs and outputs which will be introduced below), is then governed by the co-cycle $U_t$ which satisfies equation  \eqref{definition Ut} where the $SLH$-triple $(S,L,H)$ is given by Definition \ref{definition SLH composition}. It is easy to see that $U_t = U_t^1 \otimes \ldots \otimes U_t^N$. 

\begin{figure}
    \centering
    \includegraphics[width=0.75\linewidth]{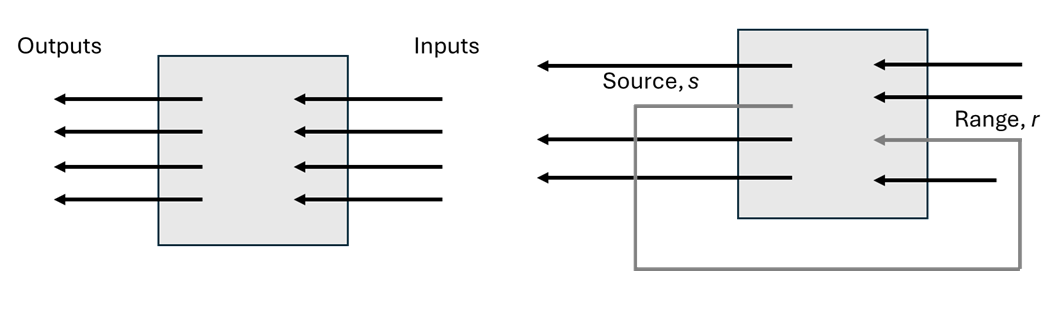}
    \caption{On the left, we have a multiple-input multiple-output system. There are the same number $m$ of inputs as outputs. An SLH network is then formed by feeding some of the outputs back in as inputs. Each connection reduces the total number of external inputs (and outputs) by 1. We take $s$ to denote the label of an output (source) and $r$ to denote the label of an input (range).}
    \label{fig:placeholder}
\end{figure}

Our next step is to create a network with connections by connecting some (not necessarily all) of the outputs to some of the inputs of the network. An output can be connected only to a single input, and \textit{vice versa}. Each feedback connection $\mathsf{c}$ is a triple $(r, s, \xi)$ of an input $r$ (\textit{range}) and an output $s$ (\textit{source}), i.e., $s, r \in \{1,\ldots,m\}$. So the output $s$ is fed back into the network as input to $r$ and the corresponding time delay is $\xi  >0$ (equal to the physical line length of the connection path divided by the speed of light). We shall denote by $\mathcal{C}$ the set of all connections in the system.  The shortest line length in the system by $\xi_{\text{min}}$ and the longest line length in the system by 
$\xi_{\text{max}}$, i.e.,
\begin{eqnarray}
    \xi_{\text{min}} =\mbox{min}\big\{ \xi ; \ (s,r,\xi) \in \mathcal{C}\big\},\qquad \xi_{\text{max}} =\mbox{max}\big\{ \xi ; \ (s,r,\xi) \in \mathcal{C}\big\}. 
\end{eqnarray}
Often, it will be convenient to assume the following:

\begin{assumption}\label{assumption multiples of tau}
All line lengths $\xi$ in the network are integer multiples of a fixed amount $\tau>0$.
\end{assumption}

\subsection{The Braid Map}
We will make use of the following convention.
For fixed connection $\mathsf{c}=(r,s,\xi )$, we place each of the labels $i=1, \cdots, m$ into one of the following four types:
\textit{
\begin{enumerate}
    \item $i$ is neither $r$ nor $s$;
    \item $r$ and $s$ are different, and $i$ equals $r$;
    \item $r$ and $s$ are different, and $i$ equals $s$;
    \item $s$ and $r$ coincide, and $i$ equals the common value.
\end{enumerate}
}
Note that the type depends on the connection $\mathsf{c}$: a label may be different types for different connections in $\mathcal{C}$.

\begin{definition}\label{definition braid map}
Let $0 < \sigma \le \xi_{\text{min}}$. 
We define a unitary map $b_\sigma(\mathsf{c}):\ L^2(\mathbb{R}; \mathbb{C}^m) \to L^2(\mathbb{R}; \mathbb{C}^m)$ by defining its vector components $\big(b_\sigma (\mathsf{c})f\big)_i$ for $i = 1,\dots, m$.
The $i$th component is defined according to which of the four types the index $i$ belongs to for the given connection $\mathsf{c}$:
\begin{enumerate}
\item if $i \neq r$ and $i \neq s$: 
\begin{equation*}
\big(b_\sigma (\mathsf{c}) f\big)_ i(t) =  f_i(t),\qquad\qquad\qquad\qquad\qquad\qquad\qquad\qquad \forall t \in \mathbb{R},
 \end{equation*}
\item if $i = r \neq s$:
\begin{equation*}
\big(b_\sigma (\mathsf{c})f\big)_{r}(t) = f_{s}\big(t-\xi\big)\chi_{[\xi,\, \xi+\sigma )}(t) + f_{r}(t) \chi_{[\xi,\, \xi+\sigma )^c}(t),\qquad \qquad\forall t \in \mathbb{R},\qquad 
\end{equation*}
\item if $i = s \neq r$:
\begin{equation*}
\big(b_\sigma (\mathsf{c})f\big)_{s}(t) = f_{r}\big(t + \xi\big)\chi_{[0,\sigma )}(t) + f_{s}(t)\chi_{[0,\sigma )^c}(t),\qquad\qquad\quad \forall t \in \mathbb{R}, 
\end{equation*}
\item if $i = s= r$:
\begin{equation*}\begin{split}
\big(b_\sigma (\mathsf{c})f\big)_i(t) =  &f_i\big(t-\xi\big)\chi_{[\xi,\, \xi+\sigma )}(t) + f_i\big(t + \xi\big)\chi_{[0,\sigma )}(t)\\
&+\,  f_i(t)\Big(\chi_{(-\infty,0)}(t) + \chi_{[\sigma , \xi)}(t)  + \chi_{[\xi+\sigma , \infty)}(t) \Big) \qquad \forall t \in \mathbb{R}.
\end{split}\end{equation*}
\end{enumerate}
Note that $b_\sigma (\mathsf{c})$ is a unitary map on $L^2(\mathbb{R}; \mathbb{C}^m)$ and that $b_\sigma (\mathsf{c}_1)$ commutes with $b_\sigma (\mathsf{c}_2)$ for all $\mathsf{c}_1, \mathsf{c}_2 \in \mathcal{C}$. We 
now define the \textbf{one-particle braid map} on $L^2(\mathbb{R}; \mathbb{C}^m)$ as
\begin{equation*}
b_\sigma  = \prod_{\mathsf{c}\in \mathcal{C}} b_\sigma  (\mathsf{c}) .
\end{equation*}
Furthermore, we may include the $\sigma =0$ case by taking $b_0$ to be the identity map on $ L^2(\mathbb{R})$. 
For $0\le \sigma \le \xi_{\text{min}}$, let $B_\sigma : \mathcal{F} \to \mathcal{F}$ be the second-quantisation of $b_\sigma $, and this defines a unitary. We call $B_\sigma $ the \textbf{braid map}. We naturally extend $B_\sigma $ to $\mathcal{H}\otimes \mathcal{F}$, by ampliation with the identity on $\mathcal{H}$.
\end{definition}

\subsection{Time evolution of the network}
For $z \in \mathbb{R}$, we shall denote by $\lfloor z\rfloor$ the integer part of $z$, i.e., the largest element in $\mathbb{N}$ that is smaller than or equal to $z$.

\begin{definition}\label{definition time evolution network}
Let $0 < \sigma \le \xi_{\text{min}}$. We define a unitary cocycle $\mathcal{U}_t$  on 
$\mathcal{H}\otimes \mathcal{F}$ by
\begin{equation*}
\mathcal{U}_t = \Theta_{\lfloor \tfrac{t}{\sigma } \rfloor \sigma }^* B_{t- \lfloor \tfrac{t}{\sigma } \rfloor \sigma }  U_{t- \lfloor \tfrac{t}{\sigma } \rfloor \sigma } \Theta_{\lfloor \tfrac{t}{\sigma } \rfloor \sigma }\overleftarrow{\prod_{i = 1}^{\lfloor \tfrac{t}{\sigma } \rfloor}}\Theta_{(i-1)\sigma }^* B_\sigma   U_\sigma  \Theta_{(i-1)\sigma }.
\end{equation*}
Note that the definition of $\mathcal{U}_t$ does not depend on which $\sigma  \in (0,\xi_{\text{min}}]$ has been used.
\end{definition}
Note that $\mathcal{U}_t$ is the co-cycle that governs the time evolution of the connected network. The physical time evolution of the network is given by the following one-parameter group of unitaries
\begin{equation}
\label{definition CalUhat}
\hat{\mathcal{U}}_t =  \left\{ \begin{array}{ll}
 \Theta_t \mathcal{U}_t & \mbox{\ \ \ if \ } t \ge 0 \\
 \mathcal{U}^*_{-t} \Theta_t & \mbox{\ \ \ if \ } t < 0 
  \end{array}\right. .
\end{equation}


\section{Convergence theorem for $SLH$-networks}\label{section convergence theorem}

In the previous section, we have seen that an $SLH$-network is completely determined by its $SLH$-triple: $(S,L,H)$ and its set of connections $\mathcal{C}$. Below, we will introduce a discrete repeated interaction model and we will show under which conditions  such a model converges to the time evolution of equation \eqref{definition Ut}. The convergence follows from \cite[Thm 1]{BvH08} in which a central role is played by the generator associated to the network. 

\begin{definition}\label{definition generator}
Let $\alpha$ and $\beta$ be elements of  $\mathbb{C}^m$ and let $(S,L,H)$ be an $SLH$-triple  with multiplicity $m$. The generator $\mathscr{L}^{\alpha\beta}: \mathcal{H} \to \mathcal{H}$ associated to a network with $SLH$-triple $(S,L,H)$ is given by
\begin{equation*}\begin{split}
    \mathscr{L}^{\alpha\beta}  = \sum_{i,j =1}^m \alpha^*_j S_{ij}^* \beta_i  +
                                                 \sum_{i=1}^m \beta_i L_i^*  -
                                                  \sum_{i,j=1}^m  \alpha_j^*S_{ij}^*L_i 
                                                 -\frac{\sum_{i=1}^m\big(|\alpha|^2_i + |\beta|^2_i + L^*_iL_i \big)}{2} +iH.
\end{split}\end{equation*}
\end{definition}

Let $\tau >0$ be the length of some time interval. If Assumption \ref{assumption multiples of tau} is in force, we will take $\tau$ from that Assumption, otherwise we can choose it freely. For $k\in \mathbb{N}$, we define
\begin{equation}\label{equation Fk} 
 \mathcal{F}^k = \mathcal{F}_{[0,\,\tau /2^k)}.
\end{equation}
Using the continuous tensor product structure of the Fock space, we can write
\begin{equation}
    \mathcal{F} =  \bigotimes_{l \in \mathbb{Z}} \mathcal{F}_{\big[ \frac{l\, \tau}{2^k},\, \frac{(l+1)\tau}{2^k} \big)} = \bigotimes_\mathbb{Z} \mathcal{F}^k,
\end{equation}
since $\mathcal{F}^k = \mathcal{F}_{[0,\,\tau \, 2^{-k})}$
is isomorphic to $\mathcal{F}_{[\, l \tau 2^{-k},\, (l+1)\tau 2^{-k})}$
for every $l \in \mathbb{Z}$ in a canonical way. Let $\mathcal{H}^k$ be a subspace of the 
initial space $\mathcal{H}$. $\mathcal{H}^k$ represents the initial space in the discretisation, i.e.\ as $k$ goes to infinity $\mathcal{H}^k$ will converge to $\mathcal{H}$. Often, however, $\mathcal{H}^k$ will simply be equal to $\mathcal{H}$ for all $k$ in our examples. Let $\mathcal{K}^k$ be a subspace of $\mathcal{F}^k$. $\mathcal{K}^k$ will act as the \emph{discrete noise space} in the following. 

Note that $\mathcal{H}^k\otimes \mathcal{K}^k$ is a subspace of $\mathcal{H}\otimes\mathcal{F}^k$. Later on, we will see that $\mathcal{H}^k\otimes\mathcal{K}^k$ is the range of the embedding of our quantum computer 
Hilbert space into $\mathcal{H}\otimes \mathcal{F}^k \subset \mathcal{H}\otimes\mathcal{F}$. 
We have chosen to introduce the discrete systems as a subsystem of our continuous description because it allows us to not carry around the notation associated with the 
embeddings in the formulation of the convergence theorems that will follow below.

\begin{definition}\label{definition discrete evolution}
Let $R^k$ be a unitary on $\mathcal{H}^k\otimes \mathcal{K}^k$ 
that describes the interaction between the initial system and the discrete 
noise in a single time step. We extend $R^k$ to $\mathcal{H}\otimes\mathcal{F}$ 
by setting it equal to the identity on the orthogonal 
complement of $\mathcal{H}^k\otimes\mathcal{K}^k$.
We now define recursively
\begin{equation}\begin{split}\label{definition Rkt}
      R^k_t{ } &= I ,\qquad \qquad\qquad\qquad\qquad\qquad\qquad \ t\in [0,\, \tfrac{\tau}{2^k}),\\
      R^k_t{ } &=   \left(\Theta^*_{\tfrac{(l-1)\tau}{2^k}}R^k\Theta_{\tfrac{(l-1)\tau}{2^{k}}}\right)R^k_{\tfrac{(l-1)\tau}{2^k}}, \qquad
                    t\in [\tfrac{l\tau}{2^k},\, \tfrac{(l+1)\tau}{2^k}),\qquad l\in \mathbb{N}\backslash \{0\}.
\end{split}\end{equation}
Furthermore, for $\psi,\varphi \in \mathcal{F}^k$, we define $R^{k;\psi,\varphi}: \mathcal{H}\to\mathcal{H}$ such that
\begin{equation}\label{equation Rkpsiphi}
	\langle u,R^{k;\psi\varphi}v\rangle = \frac{\langle u\otimes\psi, R^{k*}\,v\otimes\varphi\rangle}{\|\psi\|\,\|\varphi\|} \qquad
	\forall\,u,v\in\mathcal{H}.
\end{equation}
\end{definition}
Note that $\Theta^*_{(l-1)\tau 2^{-k}}R^k\Theta_{(l-1)\tau 2^{-k}}$ can only act 
non-trivially on $\mathcal{H} \otimes \mathcal{F}_{[ (l-1)\tau2^{-k},\, l\tau 2^{-k})}$. (That is, equation \eqref{definition Rkt} defines a repeated interaction of 
the initial system with consecutive time-slices of the field.)

We are now ready for the convergence theorem for 
the case of $SLH$-networks with no connections, i.e., $\mathcal{C}=\emptyset$. The 
theorem is a special case of \cite[Thm 1]{BvH08} in the 
sense that we allow only for bounded operators in the $SLH$-triple. 

\begin{theorem}{\bf \cite[Thm 1]{BvH08}}\label{theorem BoutenvanHandel}
The following conditions are equivalent.
\begin{enumerate}
\item For all $\alpha,\beta\in\mathbb{C}^m$,
$u \in \mathcal{H}$ there exist $u^{k}\in\mathcal{H}$
and $\psi^{k},\varphi^{k}\in\mathcal{F}^k$ such that
$$
	u^{k}\xrightarrow{k\to\infty}u,\qquad
	(\psi^{k})^{\otimes 2^k}\xrightarrow{k\to\infty}
	e(\alpha \chi_{[0,\tau)}),\qquad
	(\varphi^{k})^{\otimes 2^k}\xrightarrow{k\to\infty}
	e(\beta \chi_{[0,\tau)}),
$$
and
$$
	\frac{2^k}{\tau}(R^{k;\psi^k\varphi^k}-I)u^k
	\xrightarrow{k\to\infty}\mathscr{L}^{\alpha\beta}u.
$$
\item
For every $\psi\in\mathcal{H}\otimes\mathcal{F}$
$$
	\lim_{k\to\infty}	\|R_t^{k}\psi-U_t\psi\|=0.
$$
\item
For every $T<\infty$ and $\psi\in\mathcal{H}\otimes\mathcal{F}$
$$
	\lim_{k\to\infty}\sup_{0\le t\le T}
	\|R_t^{k}\psi-U_t\psi\|=0.
$$
\end{enumerate}
\end{theorem}

\begin{proof}
See \cite{BvH08}. Condition 2 is not explicitly stated in \cite{BvH08}. Inspection of the proof, however, shows that the proof of the implication $1 \Rightarrow 3$ goes via condition 2. The proof of the implication $3 \Rightarrow 2$ is trivial. Theorem 1 in $\cite{BvH08}$ is only for $\tau =1$. However, inspection of the proof in \cite{BvH08} shows that the rescaling to handle a general $\tau>0$ is trivial.
\end{proof}

We now turn to the situation where we do have connections in the network, i.e., $\mathcal{C} \neq \emptyset$. Recall that $\xi_{\text{min}}$ denotes the shortest line length in the network. 

\begin{definition}\label{definition discrete evolution network}
Assume that Assumption \ref{assumption multiples of tau} holds. Let $0 < \sigma \le \xi_{\text{min}}$ such that $\sigma $ is an integer multiple of $\tfrac{\tau}{2^k}$. For every $k \in \mathbb{N}\backslash \{0\}$, 
we define a unitary $\mathcal{R}^k_t$ on 
$\mathcal{H}\otimes \mathcal{F}$ by
\begin{equation*}
\mathcal{R}_t^k = \Theta_{\lfloor \tfrac{t}{\sigma } \rfloor \sigma }^* B_{t- \lfloor \tfrac{t}{\sigma } \rfloor \sigma }  R^k_{t- \lfloor \tfrac{t}{\sigma } \rfloor \sigma } \Theta_{\lfloor \tfrac{t}{\sigma } \rfloor \sigma }\overleftarrow{\prod_{i = 1}^{\lfloor \tfrac{t}{\sigma } \rfloor}}\Theta_{(i-1)\sigma }^* B_\sigma  R^k_\sigma  \Theta_{(i-1)\sigma }.
\end{equation*}
Note that the definition of $\mathcal{R}^k_t$ does not depend on which integer multiple of $\tfrac{\tau}{2^k}$ was used for $\sigma \in (0,\xi_{\text{min}}]$.
\end{definition}

We are now ready to state our main result:

\begin{theorem} \label{theorem main result}
Assume that Assumption \ref{assumption multiples of tau} holds. The following conditions are equivalent.
\begin{enumerate}
\item For all $\alpha,\beta\in\mathbb{C}^m$,
$u \in \mathcal{H}$ there exist $u^{k}\in\mathcal{H}$
and $\psi^{k},\varphi^{k}\in\mathcal{F}^k$ such that
$$
	u^{k}\xrightarrow{k\to\infty}u,\qquad
	(\psi^{k})^{\otimes 2^k}\xrightarrow{k\to\infty}
	e(\alpha \chi_{[0,\tau)}),\qquad
	(\varphi^{k})^{\otimes 2^k}\xrightarrow{k\to\infty}
	e(\beta \chi_{[0,\tau)}),
$$
and
$$
	\frac{2^k}{\tau}(R^{k;\psi^k\varphi^k}-I)u^k
	\xrightarrow{k\to\infty}\mathscr{L}^{\alpha\beta}u.
$$
\item
For every $\psi\in\mathcal{H}\otimes\mathcal{F}$
$$
	\lim_{k\to\infty}
	\|\mathcal{R}_t^{k}\psi-\mathcal{U}_t\psi\|=0.
$$
\item
For every $T<\infty$ and $\psi\in\mathcal{H}\otimes\mathcal{F}$
$$
	\lim_{k\to\infty}\sup_{0\le t\le T}
	\|\mathcal{R}_t^{k}\psi-\mathcal{U}_t\psi\|=0.
$$
\end{enumerate}
\end{theorem}

\begin{proof}
See section \ref{section proofs}.
\end{proof}
\vline

\begin{corollary}\label{corollary}
Let us return to the situation in Section \ref{section SLH-networks} where $N$ 
components with $SLH$-triples $(S^j,L^j,H^j)$ for $j =1,\ldots,N$ and associated 
unitary cocycles $U^j_t$ for $j =1,\ldots,N$ make up a network with $SLH$-triple $(S,L,H)$, 
given by Definition \ref{definition SLH composition}, 
connection set $\mathcal{C}$, and associated unitary cocycle $U_t$.  
Let $R^{k;j}$ be 
the one-time-step unitary of 
Definition \ref{definition discrete evolution}
for the separate components, $j=1,\ldots,N$. Let $R^{k;j}_t$ be 
defined for the separate components as in equation  \eqref{definition Rkt}  of 
Definition \ref{definition discrete evolution}.
Suppose that we can 
prove, e.g.\ by proving condition 1 of Theorem \ref{theorem BoutenvanHandel} and 
using the implication $1 \Rightarrow 2$ of Theorem \ref{theorem BoutenvanHandel}, that
\begin{equation*}
	\lim_{k\to\infty}
	\|R_t^{k;j}\psi-U^j_t\psi\|=0, \qquad \forall \psi \in \mathcal{H}_j\otimes\mathcal{F}_j,\qquad j = 1,\ldots,N.
\end{equation*}
Since $U_t = U^1_t\otimes\ldots\otimes U^N_t$ and $R_t^{k} = R_t^{k;1}\otimes \ldots\otimes R_t^{k;N}$, 
we immediately have 
\begin{equation*}
\lim_{k\to\infty}
	\|R_t^{k}\psi-U_t\psi\|=0, \qquad \forall \psi \in \mathcal{H}\otimes\mathcal{F}.
\end{equation*}
We can now use the implication $2\Rightarrow 1$ of Theorem \ref{theorem BoutenvanHandel}, 
followed by the implication $1 \Rightarrow 3$ of Theorem \ref{theorem main result} to 
find for every connection set $\mathcal{C}$ and $T <\infty$
\begin{equation*}
	\lim_{k\to\infty}\sup_{0\le t\le T}
	\|\mathcal{R}_t^{k}\psi-\mathcal{U}_t\psi\|=0, \qquad \forall \psi \in \mathcal{H}\otimes\mathcal{F}.
\end{equation*}
That is, convergence of the components implies convergence of the composite network 
for every possible set of connections for which Assumptions \ref{assumption multiples of tau} holds.
\end{corollary}

The goal of the simulation is to approximate the final state of the $SLH$-network after $t$ seconds of time evolution in the quantum computer. This means we also need to approximate the initial state of the network in the quantum computer. In our simulations 
we will only use coherent initial states, i.e., exponential vectors normalised to unity: $\psi(f) = \exp(-\tfrac{1}{2}\|f\|^2)e(f)$ for some $f\in L^2(\mathbb{R}; \mathbb{C}^m)$. Other states such as positive temperature states or squeezed states can also be treated in this framework by going to a different representation of the field algebra of observables, see e.g.\ \cite{Pet90}. The following Lemma shows how to approximate coherent states. 

\begin{lemma}\label{lemma converg ex}
Let $\mathsf{N} \in \mathbb{N}$. Let $f$ be an element of $L^2\big([0,\mathsf{N}\tau ]; \mathbb{C}^m\big)$. Let $f^0, f^1, f^2,\ldots$ be a sequence of elements in $L^2\big([0,\mathsf{N}\tau]; \mathbb{C}^m\big)$ such that:
\begin{enumerate}
\item for each $k\in \mathbb{N}:\ f^k$ is a simple function with respect to the $k^{th}$ 
dyadic mesh, i.e., there exist $\alpha^k_1,\alpha^k_2,\dots,\alpha^k_{\mathsf{N}2^k} \in \mathbb{C}^m$ such that
\begin{equation}\label{eq fk}
f^k =  \sum_{l=1}^{\mathsf{N}2^k} \alpha^k_l \chi_{\left[\frac{(l-1)\tau}{ 2^k},\, \frac{l\,\tau}{2^k}\right)},
\end{equation}
\item 
\begin{equation*}
   \left\| f- f^k\right\|_2  \xrightarrow{k\to\infty} 0.
\end{equation*}
\end{enumerate}
If we define
\begin{equation*}
  \psi^k_ l= 1 \oplus \alpha^k_l
                  \chi_{\left[\frac{(l-1)\tau}{2^k},\, \frac{l\tau}{2^k}\right)},\qquad  l = 1,2\ldots,\mathsf{N}2^k,  
\end{equation*}
then 
\begin{equation*}
  \left\|e(f) -\psi^k_1\otimes\psi^k_2\otimes\ldots\otimes\psi^k_{\mathsf{N}2^k} \right\|\xrightarrow{k\to\infty} 0.
\end{equation*}
\end{lemma}
\begin{proof}
Note that
\begin{equation*}\begin{split}
&\big\|e(f^k)- \psi^k_1\otimes\psi^k_2\otimes\ldots\otimes\psi^k_{\mathsf{N}2^k}\big\|^2 = \\
&\Big\langle e(f^k)- \psi^k_1\otimes\psi^k_2\otimes\ldots\otimes\psi^k_{\mathsf{N}2^k},\,
e(f^k)- \psi^k_1\otimes\psi^k_2\otimes\ldots\otimes\psi^k_{\mathsf{N}2^k}\Big\rangle = \\ 
&\left\|e(f^k)\right\|^2 + \left\|\psi^k_1\otimes\psi^k_2\otimes\ldots\otimes\psi^k_{\mathsf{N}2^k}\right\|^2 - 
2 \mbox{Re}\Big\langle e(f^k),\,\psi^k_1\otimes\psi^k_2\otimes\ldots\otimes\psi^k_{\mathsf{N}2^k}\Big\rangle   = \\ 
&\exp\left(\sum_{l=1}^{\mathsf{N}2^k}\|\alpha^k_l\|^2\frac{\tau}{2^k}\right) + \prod_{l=1}^{\mathsf{N}2^k}\left(1+ \|\alpha^k_l\|^2\frac{\tau}{2^k}\right)
- 2\prod_{l=1}^{\mathsf{N}2^k}\left(1+ \|\alpha^k_l\|^2\frac{\tau}{2^k}\right) \xrightarrow{k\to\infty} 0.
\end{split}\end{equation*}
Furthermore, $e(f)$ depends continuously on $f$ \cite[Proof Prop 19.6]{Par92}, i.e. $\|e(f)-e(f^k)\| \xrightarrow{k\to\infty} 0$. 
This means that we can now conclude
\begin{equation*}
\left\|e(f)- \psi^k_1\otimes\psi^k_2\otimes\ldots\otimes\psi^k_{\mathsf{N}2^k}\right\| \le  
\left\| e(f) - e(f^k)\right\| + \left\|e(f^k) -\psi^k_1\otimes\psi^k_2\otimes\ldots\otimes\psi^k_{\mathsf{N}2^k}\right\|
\xrightarrow{k\to\infty} 0.
\end{equation*}
\end{proof}

Note that if we take $\alpha_l^k = \displaystyle \frac{2^k}{\tau}\int_{ \frac{(l-1)\tau} {2^k} }^{\frac{l\tau}{2^k}}f(s)\,ds$, then $f^k$ as defined in 
equation  \eqref{eq fk} satisfies the second condition $\left\| f- f^k\right\|_2  \xrightarrow{k\to\infty} 0$. That is: a sequence 
$f^0,f^1,f^2,\ldots$ with the properties stated in Lemma \ref{lemma converg ex} exists.


\section{Technical intermezzo: unbounded coefficients}\label{section intermezzo}

In this technical intermezzo we are going to investigate a different 
setting than in the rest of the article\footnote{Readers that are not 
primarily interested in the unbounded are advised to skip this section 
as this is of no consequence for the following sections.}. Here, we 
will relax the condition that the operators of the $SLH$-triple 
are bounded. To this end we assume the following:

\begin{assumption}\label{assumption domains}
$\mathcal{D} \subset \mathcal{H}$ is dense in $\mathcal{H}$. The operators 
$S_{ij},\ 1\le i,j \le m$, $L_i,\ 1\le i\le m$ and $H$ are defined on $\mathcal{D}$ and 
are such that the following quantum stochastic differential equation has a unique 
unitary co-cycle solution
\begin{equation}\begin{split}\label{equation Utstar}
&dU_t^* = U_t^* \left\{\sum_{i,j =1}^m (S_{ji}^* -\delta_{ij})d\Lambda^{ij}_t - \sum_{i,j =1}^m S_{ij}^*L_i dA_t^{j*} + 
\sum_{i=1}^m L_i^*dA^i_t - \frac{1}{2}\sum_{i=1}^m L_i^*L_idt + iHdt\right\},\\
&U^*_0 = I,
\end{split}\end{equation} 
which is strongly continuous in $t$ and which extends to $\mathcal{H}$.
\end{assumption}

\begin{remark}
The reason for changing from $U_t$ to $U_t^*$ is that 
in the equation for $U_t^*$, equation \eqref{equation Utstar}, 
$U_t^*$ is on the left of the coefficients, i.e.\ the coefficients 
act immediately on vectors. This makes it generally much easier to prove existence 
of a unitary co-cycle that satisfies the QSDE \eqref{equation Utstar}, see \cite{Fag90} \cite{LiW06}.
\end{remark}

We have the following connection with the generator from 
Definition \ref{definition generator}. The proof can be found 
in \cite{BHS08}. Lemma \ref{lemma gens} also makes it clear why we used 
the hitherto unexplained $R^{k*}$ instead 
of $R^{k}$ in equation \eqref{equation Rkpsiphi}.

\begin{lemma}
\label{lemma gens}
For $\alpha,\beta\in\mathbb{C}^m$, define $T_t^{\alpha\beta}:
\mathcal{H}\to\mathcal{H}$ such that
\begin{equation*}
        \left\langle u,T_t^{\alpha\beta}v\right\rangle = e^{-\sum_{i=1}^m(|\alpha_i|^2+
                |\beta_i|^2)t/2}
        \left\langle u\otimes e(\alpha \chi_{[0,t]}),
        U^*_t\,v\otimes e(\beta \chi_{[0,t]})\right\rangle,
        \qquad \forall\,u,v\in\mathcal{H},~t\ge 0.
\end{equation*}
Then $T_t^{\alpha\beta}$ is a strongly continuous contraction
semigroup on $\mathcal{H}$, and the generator
$\mathscr{L}^{\alpha\beta}$ of this semigroup satisfies
$\mathrm{Dom}(\mathscr{L}^{\alpha\beta})\supset\mathcal{D}$
such that for $u\in\mathcal{D}$
\begin{equation*}
      \mathscr{L}^{\alpha\beta}u  = \left(\sum_{i,j =1}^m \alpha^*_j S_{ij}^* \beta_i  +
                                                 \sum_{i=1}^m \beta_i L_i^*  -
                                                  \sum_{i,j=1}^m  \alpha_j^*S_{ij}^*L_i 
                                                 -\frac{\sum_{i=1}^m\big(|\alpha|^2_i + |\beta|^2_i + L^*_iL_i \big)}{2} +iH\right)u.
\end{equation*}
\end{lemma}

Definition \ref{definition discrete evolution} 
remains exactly the same as before.  Naturally, we take $U_t = (U_t^*)^*$ with $U_t^*$ given 
by equation \eqref{equation Utstar}.
We can now formulate the unbounded version of Theorem \ref{theorem BoutenvanHandel}:

\begin{theorem}{\bf \cite[Thm 1]{BvH08}}\label{theorem BoutenvanHandel unbounded}
Assume that Assumption \ref{assumption domains} holds, and let 
$\mathcal{D}^{\alpha\beta}\subset 
\mathrm{Dom}(\mathscr{L}^{\alpha\beta})$ be a core for
$\mathscr{L}^{\alpha\beta}$, $\alpha,\beta\in\mathbb{C}^m$.  Then the
following conditions are equivalent.
\begin{enumerate}
\item For all $\alpha,\beta\in\mathbb{C}^m$,
$u\in\mathcal{D}^{\alpha\beta}$ there exist $u^{k}\in\mathcal{H}$
and $\psi^{k},\varphi^{k}\in\mathcal{F}^k$ such that
$$
	u^{k}\xrightarrow{k\to\infty}u,\qquad
	(\psi^{k})^{\otimes 2^k}\xrightarrow{k\to\infty}
	e(\alpha \chi_{[0,\tau)}),\qquad
	(\varphi^{k})^{\otimes 2^k}\xrightarrow{k\to\infty}
	e(\beta \chi_{[0,\tau)}),
$$
and
$$
	\frac{2^k}{\tau}(R^{k;\psi^k\varphi^k}-I)u^k
	\xrightarrow{k\to\infty}\mathscr{L}^{\alpha\beta}u.
$$
\item
For every $\psi\in\mathcal{H}\otimes\mathcal{F}$
$$
	\lim_{k\to\infty}	\|R_t^{k}\psi-U_t\psi\|=0.
$$
\item
For every $T<\infty$ and $\psi\in\mathcal{H}\otimes\mathcal{F}$
$$
	\lim_{k\to\infty}\sup_{0\le t\le T}
	\|R_t^{k}\psi-U_t\psi\|=0.
$$
\end{enumerate}
\end{theorem}

\begin{proof}
See \cite{BvH08}. Condition 2 is not explicitly stated in \cite{BvH08}. Inspection of the proof, however, shows that the proof of the implication $1 \Rightarrow 3$ goes via condition 2. The proof of the implication $3 \Rightarrow 2$ is trivial. Theorem 1 in $\cite{BvH08}$ is only for $\tau =1$. However, inspection of the proof in \cite{BvH08} shows that the rescaling to handle a general $\tau>0$ is trivial.
\end{proof}

We leave Definition \ref{definition braid map}, Definition \ref{definition discrete evolution} 
and Definition \ref{definition discrete evolution network} exactly as before.
We use $U_t = (U_t^*)^*$ in Definition \ref{definition time evolution network}, where $U_t^*$ is given 
by equation \eqref{equation Utstar}. Apart 
from that, however, we leave Definition \ref{definition time evolution network} exactly as before.
We are now ready to generalize our main result to the unbounded case:

\begin{theorem} \label{theorem main result unbounded}
Assume that Assumptions \ref{assumption multiples of tau} and \ref{assumption domains} hold, and let 
$\mathcal{D}^{\alpha\beta}\subset 
\mathrm{Dom}(\mathscr{L}^{\alpha\beta})$ be a core for
$\mathscr{L}^{\alpha\beta}$, $\alpha,\beta\in\mathbb{C}^m$.  The
following conditions are equivalent. 
\begin{enumerate}
\item For all $\alpha,\beta\in\mathbb{C}^m$,
$u\in\mathcal{D}^{\alpha\beta}$ there exist $u^{k}\in\mathcal{H}$
and $\psi^{k},\varphi^{k}\in\mathcal{F}^k$ such that
$$
	u^{k}\xrightarrow{k\to\infty}u,\qquad
	(\psi^{k})^{\otimes 2^k}\xrightarrow{k\to\infty}
	e(\alpha \chi_{[0,\tau)}),\qquad
	(\varphi^{k})^{\otimes 2^k}\xrightarrow{k\to\infty}
	e(\beta \chi_{[0,\tau)}),
$$
and
$$
	\frac{2^k}{\tau}(R^{k;\psi^k\varphi^k}-I)u^k
	\xrightarrow{k\to\infty}\mathscr{L}^{\alpha\beta}u.
$$
\item
For every $\psi\in\mathcal{H}\otimes\mathcal{F}$
$$
	\lim_{k\to\infty}
	\|\mathcal{R}_t^{k}\psi-\mathcal{U}_t\psi\|=0.
$$
\item
For every $T<\infty$ and $\psi\in\mathcal{H}\otimes\mathcal{F}$
$$
	\lim_{k\to\infty}\sup_{0\le t\le T}
	\|\mathcal{R}_t^{k}\psi-\mathcal{U}_t\psi\|=0.
$$
\end{enumerate}
\end{theorem}

\begin{proof}
See section \ref{section proofs}. The proof is the same as 
the proof of Theorem \ref{theorem main result}, but whenever 
Theorem \ref{theorem BoutenvanHandel} is used in the proof, 
one needs to invoke the unbounded version: Theorem \ref{theorem BoutenvanHandel unbounded}.
\end{proof}

\begin{remark} {\bf \cite[Remark 4]{BvH08}}
Theorem \ref{theorem BoutenvanHandel unbounded} 
and Theorem \ref{theorem main result unbounded} 
are the most powerful when $\mathcal{D}$ is a core
for all $\mathscr{L}^{\alpha\beta}$, $\alpha,\beta\in\mathbb{C}^m$.  We can 
then choose $\mathcal{D}^{\alpha\beta}\subset\mathcal{D}$, with the 
important consequence that this puts the explicit expression for 
$\mathscr{L}^{\alpha\beta}$ in Lemma \ref{lemma gens} at our disposal.  
Typically existence and uniqueness
proofs for the solution of equation \eqref{equation Utstar} already imply that 
$\mathcal{D}$ is a core for $\mathscr{L}^{\alpha\beta}$, see, e.g.,
\cite{Fag90,LiW06} and \cite[remark 4]{BHS08} for further comments.
In connection to \cite[Remark 4]{BvH08}, we remark that a 
careful reading of the proof of Theorems \ref{theorem main result}
and \ref{theorem main result unbounded} 
reveals that it only depends on the fact that $U_t$ is a cocycle with 
respect to the shift just like \cite[Thm 1]{BvH08}. Therefore, Theorem \ref{theorem main result unbounded}, just like 
\cite[Thm 1]{BvH08}, could be expressed in even greater generality. 
\end{remark}

\begin{corollary}
Suppose we have $N$ components, all satisfying Assumption \ref{assumption domains}. 
This leads to  $N$ unitary cocycles $U^1_t,\dots, U^N_t$,  where $U^j_t := (U^{j*}_t)^*$ for $ 1\le i \le N$).
We define a unitary cocyle $U_t := U^1_t \otimes \ldots \otimes U^N_t$. 
Let $R^{k;j}$ be the one-time-step unitary of 
Definition \ref{definition discrete evolution}
for the separate components, $j=1,\ldots,N$. Let $R^{k;j}_t$ be 
defined for the separate components, $j=1,\ldots,N$, as in equation  \eqref{definition Rkt}  of 
Definition \ref{definition discrete evolution}.
Suppose that we can 
prove, e.g.\ by proving condition 1 of Theorem \ref{theorem BoutenvanHandel unbounded} and 
using the implication $1 \Rightarrow 2$ of Theorem \ref{theorem BoutenvanHandel unbounded}, that
\begin{equation*}
	\lim_{k\to\infty}
	\|R_t^{k;j}\psi-U^j_t\psi\|=0, \qquad \forall \psi \in \mathcal{H}_j\otimes\mathcal{F}_j,\qquad j = 1,\ldots,N.
\end{equation*}
Since $U_t = U^1_t\otimes\ldots\otimes U^N_t$ and $R_t^{k} = R_t^{k;1}\otimes \ldots\otimes R_t^{k;N}$, 
we immediately have 
\begin{equation*}
\lim_{k\to\infty}
	\|R_t^{k}\psi-U_t\psi\|=0, \qquad \forall \psi \in \mathcal{H}\otimes\mathcal{F}.
\end{equation*}
We use the implication $2\Rightarrow 1$ of Theorem \ref{theorem BoutenvanHandel unbounded}
with $\mathcal{D}^{\alpha\beta} = \mathrm{Dom}(\mathscr{L}^{\alpha\beta})$, for $\alpha,\beta\in\mathbb{C}^m$, 
followed by the implication $1 \Rightarrow 3$ of Theorem \ref{theorem main result unbounded} 
with $\mathcal{D}^{\alpha\beta} = \mathrm{Dom}(\mathscr{L}^{\alpha\beta})$, for $\alpha,\beta\in\mathbb{C}^m$, 
to find for every connection set $\mathcal{C}$ and $T <\infty$
\begin{equation*}
	\lim_{k\to\infty}\sup_{0\le t\le T}
	\|\mathcal{R}_t^{k}\psi-\mathcal{U}_t\psi\|=0, \qquad \forall \psi \in \mathcal{H}\otimes\mathcal{F},
\end{equation*}
where $\mathcal{R}_t^{k}$ is defined by Definition \ref{definition discrete evolution network} and $\mathcal{U}_t$ is defined 
via Definition \ref{definition time evolution network} 
with $U_t = U^1_t\otimes\ldots\otimes U^N_t$.
That is, if Assumption \ref{assumption domains} is satisfied for all components, 
convergence of the components implies convergence of the composite network 
for every possible set of connections for which Assumption \ref{assumption multiples of tau} holds .
\end{corollary}


\section{Simulation of $SLH$-networks on a quantum computer}\label{section SLH on the quantum computer}

We shall emulate the SLH network with an assembly of qubits which 
we functionally split into two groups: the hardware 
(with underlying quantum computer Hilbert space $\mathbb{H}$); 
and the noise/signals (with underlying one-time-step quantum computer Hilbert space $\mathbb{K}$).

Let $\mathbb{H} = (\mathbb{C}^2)^{\otimes p}$ be the $p$ 
qubits in the quantum computer Hilbert space with which we describe 
the initial system of the network. Note that in principle we can let $p$  
depend on $k$, i.e.\ as the mesh gets finer and finer, we can approximate 
the initial system better and better. Note that, 
although from this point onwards, we only consider bounded 
operators, the initial space $\mathcal{H}$ can still be infinite 
dimensional, i.e.\ it is possible that it needs to be 
approximated.  In the examples in this paper, however, we 
will always take $p$ independent of $k$.   

Let $\mathbb{K} = (\mathbb{C}^2)^{\otimes q}$ be 
$q$ further qubits in the quantum computer Hilbert space 
with which we describe the noise 
of the $m$ input-output channels in one time step.
The typical case we will be studying is when $q$ equals $m$.
Note that if $m \ge 2$, this is different from the 
situation in \cite{AtP06} which uses $m+1$ dimensional 
noise (one excited state per channel + the vacuum). We 
like to use one qubit per channel, though, (excited state + vacuum per channel) 
because later on we will connect the components together and 
the different channels will then be split up and moved to 
a following component. The separate qubits can move to 
different components in this process. Since we 
use (at least) one qubit per channel and we wish to 
to describe more than just one channel, we use
the more general results of \cite{BvH08}
throughout this article.

Further note that we could choose to describe a particular channel or channels
with more qubits in one time step than the others. 
In that way we could differentiate in the time-scales 
that are present in the different channels. Instead of 
making the time step smaller to catch the fastest time-scale (and as a result needing more 
qubits for all channels), we could keep the time step 
constant and just describe the "fast" channels more 
accurately. Naturally, in such a description 
direct scattering (via the $S_{ij}$ operators of the 
$SLH$-triple) between "fast" and "slow" channels and connections between "fast"  and "slow" channels
are forbidden. In the remainder of the article we will 
always use one qubit per channel, i.e.\ we take $q=m$. 

\begin{definition}\label{definition N}
For $x \in \mathbb{R}$, we denote by $\lfloor x \rfloor$ the 
integer part of $x$. If the connection set $\mathcal{C}$ is not empty, 
we define $\mathsf{N}_k \in \mathbb{N}$ by
\begin{equation*}
\mathsf{N}_k := \left\lfloor \frac{2^k (T + \xi_{\text{max}})}{\tau}\right\rfloor,
\end{equation*} 
where $T$ is the time horizon over which we wish 
to run the $SLH$-network. 
If the connection set $\mathcal{C}$ is 
empty, then we define  $\mathsf{N}_k := \lfloor\tfrac{2^kT}{\tau}\rfloor$.
\end{definition}

\begin{definition}\label{definition computer shift}
Let $\theta: \{1,2,\dots,m\} \times \{1,2,\ldots,\mathsf{N}_k\} \to \{1,2,\dots,m\} \times \{1,2,\ldots,\mathsf{N_k}\}$
be the permutation given by
\begin{equation*}
\theta\left(\left(i, \mathsf{N}_k\right)\right) = (i, 1), \qquad  
\theta\left( (i,j)\right) = (i, j+1), \ \ \mbox{if \ } j \neq  \mathsf{N}_k.
\end{equation*}
We define the one-time-step left shift $\Theta:\ \mathbb{K}^{\otimes \mathsf{N}_k} \to  \mathbb{K}^{\otimes \mathsf{N}_k}$ by linear extension of its action on pure tensors:
\begin{equation*}
\Theta \left(\bigotimes_{j=1}^{\mathsf{N}_k}\left(\bigotimes_{i=1}^m v_{(i,j)}\right)\right) := 
\bigotimes_{j=1}^{\mathsf{N}_k}\left(\bigotimes_{i=1}^m v_{\theta\left((i,j)\right)}\right).
\end{equation*}
\end{definition}

On $\mathbb{H}\otimes\mathbb{K}$ we introduce 
a unitary $M^k$ that represents the interaction between 
the initial system $\mathbb{H}$ and the discrete noise
$\mathbb{K}$ in the quantum computer Hilbert space 
in one time step. Note that a quantum computer has 
a set of gates that is universal, i.e.\ every unitary 
on the quantum computer space can be constructed 
as a composition of a finite number of gates from the 
set of implemented two and one qubit gates \cite{DiV95, Bar95}.
 Let us now first assume we have 
no connections in the network, i.e.\ $\mathcal{C} = \emptyset$. 
We assume that the quantum computer has  
enough qubits in the quantum computer Hilbert space 
to fit $\mathbb{H}$ and $\mathsf{N}_k$ separate 
copies of $\mathbb{K}$, i.e.\ at least $p + \lfloor \tfrac{2^kT}{\tau}\rfloor m$ 
qubits.

\begin{definition}\label{definition qprocess in qcomp space}
On the space $\mathbb{H}\otimes \mathbb{K}^{\otimes \mathsf{N}_k}$ 
we define a unitary process by $M^k(0) = I$ and for $l\in \mathbb{N}$ such that 
$1 \le l \le \mathsf{N}_k$:
\begin{equation}\label{definition Rkl}
M^k(l) = \overleftarrow{\prod}_{i = 1}^l {\Theta^*}^{i-1}M^k\Theta^{i-1} = M^k_l M^k_{l-1}\ldots M^k_2M^k_1.
\end{equation}
Here $M_i^k$ stands for the interaction $M^k$ between $\mathbb{H}$
and the $i^{\mbox{th}}$ copy of the noise space $\mathbb{K}$ in the 
tensor product $\mathbb{K}^{\otimes \mathsf{N}_k}$.
\end{definition}

\begin{definition}
Let $\pi^k$ be a sequence of bounded maps 
$\pi^k = \pi_1^k \otimes \pi_2^k$ where 
$\pi^k_1: \mathbb{H} \to \mathcal{H}$ and 
$\pi^k_2: \mathbb{K} \to \mathcal{F}^k$ which 
are all partially isometric in the following sense
\begin{equation}\label{definition pik}\begin{split}
    &(\pi^{k}_1)^\ast \pi^k_1 = I_\mathbb{H}, \quad \pi^{k}_1 (\pi^{k}_1 )^\ast = P_{\mathcal{H}^k},\quad \mbox{where}\ \mathcal{H}^k := \mbox{ran}\, \pi_1^k,\\
    &(\pi^{k}_2)^\ast\pi^k_2 = I_\mathbb{K}, \quad \pi^{k}_2 (\pi^{k}_2)^\ast = P_{\mathcal{K}^k},\quad \mbox{where}\ \mathcal{K}^k := \mbox{ran}\, \pi_2^k.
\end{split}\end{equation}
Furthermore, we define $R^k$ on $\mathcal{H}^k\otimes \mathcal{K}^k$ by
\begin{equation}\label{definition Rk}
    R^k := \pi^k M^k (\pi^{k})^\ast .
\end{equation}
\end{definition}

Note that we now have the following Lemma:

\begin{lemma}\label{lemma Rk}
For all $l \in \mathbb{N}$ such that $0 \le l \le \lfloor\tfrac{2^kT}{\tau}\rfloor$ and all $\psi \in \mathcal{H}\otimes \big( {\mathcal{K}^k}\big)^{\otimes \mathsf{N}_k}\subset 
\mathcal{H}\otimes \mathcal{F}_{[0,\, \mathsf{N}_k\tau 2^{-k}]}$
\begin{equation*}
    R_t^k\psi = \bigg( \pi_1^k\otimes \big({\pi^k_2}\big)^{\otimes \mathsf{N}_k} \bigg)M^k(l) \bigg( \pi_1^{k}\otimes \big( {\pi^{k}_2}\big) ^{\otimes \mathsf{N}_k}\bigg)^\ast \, \psi,\qquad \forall t \in \big[\tfrac{l\tau}{2^k},\, \tfrac{(l+1)\tau}{2^k}\big)\cap [0,T].
\end{equation*}
\end{lemma}
\begin{proof}
Immediate from the definitions of $R^k$ in equation \eqref{definition Rk}, $R_t^k$ in equation\ \eqref{definition Rkt}, 
$M^k(l)$ in equation \eqref{definition Rkl} and $\pi_1^k$ and $\pi_2^k$ in equation \eqref{definition pik}.
\end{proof}

To see if a repeated interaction model given by equation  \eqref{definition Rkl}, implemented on a quantum computer, converges to the solution of a quantum stochastic differential 
equation of the form given by equation  \eqref{definition Ut} for some triple $(S,L,H)$, we can now simply check condition 1 of Theorem \ref{theorem BoutenvanHandel}. In the 
following section we will see how this works in several examples. 

Now let us assume the $SLH$-network does have connections, 
i.e.\ $\mathcal{C} \neq \emptyset$ and as a consequence $\xi_{\text{max}}> 0$. 
We assume that Assumption \ref{assumption multiples of tau} is in force.
Let $T$ again be the time horizon over which we wish to run the simulation. 
Furthermore, recall that for $k \in \mathbb{N}$, we have
due to Definition \ref{definition N} that $\mathsf{N}_k = \left\lfloor \frac{2^k (T + \xi_{\text{max}})}{\tau}\right\rfloor$.  
We here assume once more that we can describe the initial system with $p$ qubits.
We will first assume that we have a sufficiently large number of qubits to our availability:
$p + \mathsf{N}_km$ in total.
  \\
\begin{definition}\label{definition quantum computer braid}
Let $\mathsf{c} = (r,\,s,\, \xi)$ be an element of the connection set $\mathcal{C}$.
Let $\Pi^k(\mathsf{c}): \{1,2,\dots,m\} \times \{1,2,\ldots,\mathsf{N}_k\} \to \{1,2,\dots,m\} \times \{1,2,\ldots,\mathsf{N}_k\}$
be the permutation given by
\begin{equation*}\begin{split}
&\Pi^k(\mathsf{c}) \left(\left(r,\, \frac{2^k\xi}{\tau}+1\right)\right) = (s,\,1), \qquad \Pi^k(\mathsf{c}) \left(\left(s,\,1\right)\right) = \left(r,\, \frac{2^k\xi}{\tau}+1\right)  ,\\
&\Pi^k(\mathsf{c})\big((i,j)\big) = (i,j), \qquad \mbox{if \ } (i,j) \neq (s,\,1)  \mbox{\ and \ } (i,j) \neq \left(r,\, \frac{2^k\xi}{\tau}+1\right).
\end{split}\end{equation*} 
We now define the one-time-step braid map associated to $\mathsf{c}$,
$B^k(\mathsf{c}):\ \mathbb{K}^{\otimes \mathsf{N}_k} \to 
\mathbb{K}^{\otimes \mathsf{N}_k}$, 
by linear extension of the following action on pure tensors
\begin{equation*}
B^k(\mathsf{c})\left(\bigotimes_{j=1}^{\mathsf{N}_k}\left(\bigotimes_{i=1}^m v_{(i,j)}\right)\right) := 
\bigotimes_{j=1}^{\mathsf{N}_k}\left(\bigotimes_{i=1}^m v_{\Pi^k(\mathsf{c})\left((i,j)\right)}\right).
\end{equation*}
We define the one-time-step braid map $B^k:\ \mathbb{K}^{\otimes \mathsf{N}_k} \to 
\mathbb{K}^{\otimes \mathsf{N}_k}$ by
\begin{equation*}
B^k := \prod_{\mathsf{c}\in\mathcal{C}} B^k(\mathsf{c}).
\end{equation*}
\end{definition}

\begin{definition}\label{definition cal M}
On the space $\mathbb{H}\otimes \mathbb{K}^{\otimes \mathsf{N}_k}$ 
we define a unitary process by $\mathcal{M}^k(0) = I$ and for $l\in \mathbb{N}$ such that 
$1 \le l \le \lfloor\tfrac{2^kT}{\tau}\rfloor$:
\begin{equation}
\mathcal{M}^k(l) = \overleftarrow{\prod}_{i = 1}^l {\Theta^*}^{i-1}(B^kM^k)\Theta^{i-1}.
\end{equation}
\end{definition}

The following Lemma  extends Lemma \ref{lemma Rk} to the case where the 
connection set $\mathcal{C}$ is not empty.

\begin{lemma}\label{lemma CalRk}
Assume Assumption \ref{assumption multiples of tau} holds. 
For all $l \in \mathbb{N}$ such that $0 \le l \le \lfloor\tfrac{2^kT}{\tau}\rfloor$ and all 
$\psi \in \mathcal{H}\otimes \big( {\mathcal{K}^k}\big)^{\otimes \mathsf{N}_k}\subset \mathcal{H}\otimes \mathcal{F}_{[0,\, \mathsf{N}_k\tau 2^{-k}]}$, we have
\begin{equation*}
    \mathcal{R}_t^k\psi = 
\bigg( \pi_1^k\otimes \big({\pi^k_2}\big)^{\otimes \mathsf{N}_k} \bigg)
\mathcal{M}^k(l)
\bigg( \pi_1^{k}\otimes \big( {\pi^{k}_2}\big) ^{\otimes \mathsf{N}_k}\bigg)^\ast \, \psi,
\end{equation*}
for all $t \in \big[\tfrac{l\tau}{2^k},\, \tfrac{(l+1)\tau}{2^k}\big)\cap[0,T]$.
\end{lemma}
\begin{proof}
Immediate from the definitions of $R^k$ in equation \eqref{definition Rk}, 
$R_t^k$ in equation\ \eqref{definition Rkt}, 
$\mathcal{R}_t^k$ in Definition \ref{definition discrete evolution network}, 
$\pi_1^k$ and $\pi_2^k$ in equation \eqref{definition pik}, 
$\Theta$ in Definition \ref{definition computer shift}, $B^k$
in Definition \ref{definition quantum computer braid}, and $\mathcal{M}^k(l)$ in Definition \ref{definition cal M}.
\end{proof}

\begin{remark}\label{remark qubits needed}
For all channels whose inputs are connected to an output with 
connection length $\xi$, all
qubits after number $\frac{\xi 2^k}{\tau}$ will never enter the internal 
part of the network. These qubits will be transported to the output 
by the braid map without ever interacting with the system. We can 
therefore eliminate these qubits from the description.  

Furthermore, in an actual simulation it is very often the case that it is not 
needed to keep the entire output stored. Often one just 
wants to measure one observable in an output, record the result and then 
the qubit can be reinitialized and used once more as an input 
qubit \cite{BVS19}. This drastically reduces the number of needed qubits $n_k$ to 
\begin{equation*}
 n_k = p+ m - |\mathcal{C}| + \sum_{i =1}^{|\mathcal{C}|} \frac{\xi_i 2^k}{\tau} ,
\end{equation*}
where the $\xi_i$'s are all the connection lengths in the network and $|\mathcal{C}|$ 
is the size of the connection set, i.e.\ the number of connections in the network.
\end{remark}


\section{Examples: network components}\label{section examples components}

\begin{example}{\bf System with single input-output channel\ }
In the quantum computer, we identify $p$ qubits and we let
$\mathbb{H} = \left(\mathbb{C}^2\right)^{\otimes p}$ be the Hilbert space on which they are 
represented. For simplicity, we
identify the initial space $\mathcal{H}$ of the system we are studying 
with $\mathbb{H}$. That is:  $\mathcal{H} = \mathbb{H}$ and $\pi_1$ is 
the identity map. In our example, the noise space 
$\mathbb{K}^{\otimes \lfloor\frac{2^k T}{\tau}\rfloor}$ consists of a string 
of qubits (not used by the initial space $\mathbb{H}$), i.e.\ the single 
time-step noise space is a qubit: $\mathbb{K} =\mathbb{C}^2$.
We define $\pi_2: \mathbb{K} \to \mathcal{F}^k$ by 
\begin{equation*}
\pi^k_2 |0 \rangle  = 1, \qquad \pi^k_2 |1\rangle  = \sqrt{\frac{2^k}{\tau}}\chi_{\left[0,\, \frac{\tau}{2^k}\right)}, 
\end{equation*}
i.e.\ $\mathcal{K}^k = \mbox{span}\, \left\{1,  \sqrt{\frac{2^k}{\tau}}\, \chi_{\left[0,\, \frac{\tau}{2^k}\right)}\right\} \subset \mathcal{F}_k$.

The interaction $M^k:\ \mathbb{H}\otimes \mathbb{K} \to \mathbb{H}\otimes \mathbb{K}$ 
between the initial system and one time step of the discrete noise is given by:
\begin{equation}\begin{split}\label{definition M}
M^k{ } &= \exp\left(\sqrt{\frac{\tau}{2^k}}\Big(L\otimes \sigma_+ - L^*\otimes \sigma_-\Big)\right)
          \exp\left( -i \frac{\tau}{2^k} H \otimes   I\right)
        \\
&= \begin{bmatrix}
    \cos(\sqrt{\frac{\tau}{2^k}LL^*}) & \frac{\sin(\sqrt{\frac{\tau}{2^k}LL^*})}{\sqrt{LL^*}}L \\
    -  \frac{\sin(\sqrt{\frac{\tau}{2^k}L^*L})}{\sqrt{L^*L}}L^*  & \cos(\sqrt{\frac{\tau}{2^k}L^*L})    
\end{bmatrix}
\begin{bmatrix}
      \exp\left( -i \frac{\tau}{2^k} H\right) & 0 \\
            0   &  \exp\left( -i \frac{\tau}{2^k} H\right)
\end{bmatrix}.
\end{split}\end{equation}

Here the operator $L: \mathbb{H} \to \mathbb{H}$ couples the system 
to the noise and the self-adjoint operator $H: \mathbb{H} \to \mathbb{H}$ is the Hamiltonian governing
the internal evolution of the initial system.
Note that with respect to the 
basis $\left(1, \sqrt{\frac{2^k}{\tau}}\chi_{\left[0,\, \frac{\tau}{2^k}\right)}\right)$, 
the unitary operator $R^k$ (equation \eqref{definition Rk}) 
is given by the matrix of $M^k$ (see equation  \eqref{definition M}).

Let $\alpha, \beta$ be elements of $\mathbb{C}$.
 If we define
\begin{equation*}
  \psi^k = 1 \oplus \alpha\, \chi_{\left[0,\, \frac{\tau}{2^k}\right)}, \qquad 
  \varphi^k = 1 \oplus \beta\, \chi_{\left[0,\, \frac{\tau}{2^k}\right)},
\end{equation*}
then $\psi^k, \varphi^k \in \mathcal{F}^k$ and due to Lemma \ref{lemma converg ex},
we have
\begin{equation*}
(\psi^{k})^{\otimes 2^k}\xrightarrow{k\to\infty}
	e(\alpha \chi_{[0,\tau)}),\qquad
(\varphi^{k})^{\otimes 2^k}\xrightarrow{k\to\infty}
	e(\beta \chi_{[0,\tau)}).
\end{equation*}

Let $u$  be an element of $\mathcal{H} =\mathbb{H}$. We take $u^k = u$ for all $k$, and we compute
\begin{eqnarray*}
\frac{2^k}{\tau}(R^{k;\psi^k\varphi^k}-I)u^k &=&  \dfrac{1}{\sqrt{1+ \frac{|\alpha|^2\tau}{2^k}}\sqrt{1+ \frac{|\beta|^2|\tau}{2^k}}}\\
&&\times \bigg\{ e^{i \frac{\tau}{2^k}H}\Bigg(
\alpha^*\beta \cos(\sqrt{\frac{\tau}{2^k}LL^*}) - 
\alpha^*\sqrt{\frac{2^k}{\tau}}L\frac{\sin(\sqrt{\frac{\tau}{2^k}L^*L})}{\sqrt{L^*L}}    \\
&& +
\beta\sqrt{\frac{2^k}{\tau}}L^*\frac{\sin(\sqrt{\frac{\tau}{2^k}LL^*})}{\sqrt{LL^*}} +\frac{2^k}{\tau} \cos(\sqrt{\frac{\tau}{2^k}L^*L}
\Bigg)u \bigg\} -
\frac{2^k}{\tau} u \\
&&
\xrightarrow{k\to\infty}
\left(\alpha^*\beta + \beta L^* - \alpha^* L - \frac{|\alpha|^2 + |\beta|^2  + L^*L}{2} + iH\right)u,
\end{eqnarray*}

which we recognize (see Definition \ref{definition generator})
as the generator $\mathscr{L}^{\alpha \beta}$, acting on $u$, associated 
to the $SLH$-triple:
\begin{equation*}
  \left(I,\, L,\, H\right).
\end{equation*} 

Having checked all requirements for 
condition 1 of Theorem \ref{theorem BoutenvanHandel}, 
we can now conclude, by invoking the equivalence of condition 
1 and condition 3 of Theorem \ref{theorem BoutenvanHandel}, 
that (the embedding of) the repeated interaction of equation  \eqref{definition Rkl} with 
$M^k$ as in equation  \eqref{definition M} converges strongly, uniformly on 
compact time intervals, to the solution $U_t$ of the 
following QSDE
\begin{equation}\label{qsde no scat}\begin{split}
dU_t = \left\{L dA^{*}_{t} - L^* dA_t 
           -\frac{L^*L}{2}  dt - iH dt \right\}U_t,
\end{split}\end{equation}
with $U_0 = I$. Note that the discretisation given by
$M^k$ in equation \eqref{definition M} is not unique.
Many other choices can converge to the same solution 
$U_t$ of equation  \eqref{qsde no scat}. 
Note that since we only have one input-output channel 
in this example, we could also have used the result 
in \cite{AtP06} to prove the convergence. 

See \cite{ViB19} for simulations of equation  \eqref{qsde no scat}  
with $\mathbb{H} = \mathbb{C}^2$ (representing a two-level atom), 
$L = \sqrt{\kappa}\sigma_-$  ($\kappa$ is the \emph{decay rate})
and $H = \omega \sigma_+\sigma_- + \frac{\Omega}{2}\sigma_y$, where 
$\omega$ is the \emph{transition frequency} of the two-level atom 
and $\Omega$ is the frequency of the 
\emph{Rabi oscillations} which originate from the interaction of a classical 
laser field with the atomic spin (not modeled here, but represented as a 
Hamiltonian term). The simulations in \cite{ViB19} are 
on an actual digital quantum computer and are based on a discretisation that differs 
only slightly from the one given by
Eqns.\ \eqref{definition Rkl} and  \eqref{definition M}. In \cite{ViB19} one 
can also find the circuit (i.e.\ the quantum gates) that implements 
the unitary $M^k$ of equation  \eqref{definition M} on a digital quantum computer.
Furthermore, see \cite{ViB19} to see how the discrete master 
equation and discrete quantum filtering equations \cite{BHJ09} (also known as 
stochastic master equations) can be obtained from the
simulations.
\end{example}

\begin{example}{\bf System with 3 input-output channels\ }\label{example 3 channels}
We now couple the initial system $\mathcal{H} = \mathbb{H} = \mathbb{C}^p$ 
from the previous example to three channels in the field. In Examples \ref{example Eschner} and \ref{example cavity QED},  
we will use the first two 
of these three channels to have photons traveling from right to left (channel $1$) and 
from left to right (channel $2$). In Example \ref{example scattering}, we will introduce a 
discrete description of optical elements such as mirrors and beam splitters. 
We can use these elements to have photons reverse direction at the mirrors 
and to create a cavity with photons traveling in two directions 
bouncing between the mirrors. The initial system (typically an atom) will 
sit inside the cavity and will interact with the photons travelling in either direction.
The third channel we will use for driving the system from the side, i.e.\ this channel is perpendicular 
to the two other channels and is in a coherent state which represents driving by a laser. 

The single time-step noise space now consists of three qubits: 
$\mathbb{K} = \mathbb{C}^2\otimes \mathbb{C}^2\otimes \mathbb{C}^2$.
We define $\pi^k_2: \mathbb{K}\to \mathcal{F}^k$ by
\begin{equation}\begin{split}\label{eqn pi2}
&\pi^k_2 |000\rangle = 1, \qquad 
\pi^k_2 |100\rangle = \sqrt{\frac{2^k}{\tau}}\begin{bmatrix}\chi_{\left[0,\, \frac{\tau}{2^k}\right)} \\ 0 \\ 0\end{bmatrix},\qquad
\pi^k_2 |010\rangle = \sqrt{\frac{2^k}{\tau}}\begin{bmatrix} 0 \\ \chi_{\left[0,\, \frac{\tau}{2^k}\right)}  \\ 0\end{bmatrix},\\
&\pi^k_2 |001\rangle = \sqrt{\frac{2^k}{\tau}}\begin{bmatrix}0 \\ 0\\ \chi_{\left[0,\, \frac{\tau}{2^k}\right)} \end{bmatrix},\qquad
\pi^k_2 |110\rangle = \frac{2^k}{\tau}\begin{bmatrix}\chi_{\left[0,\, \frac{\tau}{2^k}\right)} \\ 0 \\ 0\end{bmatrix}\otimes\begin{bmatrix} 0 \\ \chi_{\left[0,\, \frac{\tau}{2^k}\right)}  \\ 0\end{bmatrix},\\
&\pi^k_2 |101\rangle = \frac{2^k}{\tau}\begin{bmatrix}\chi_{\left[0,\, \frac{\tau}{2^k}\right)} \\ 0 \\ 0\end{bmatrix}\otimes\begin{bmatrix} 0 \\ 0 \\ \chi_{\left[0,\, \frac{\tau}{2^k}\right)} \end{bmatrix},\qquad
\pi^k_2 |011\rangle = \frac{2^k}{\tau}\begin{bmatrix} 0 \\ \chi_{\left[0,\, \frac{\tau}{2^k}\right)} \\ 0 \end{bmatrix}\otimes\begin{bmatrix} 0 \\ 0 \\\chi_{\left[0,\, \frac{\tau}{2^k}\right)} \end{bmatrix},\\
&\pi^k_2 |111\rangle = \frac{2^k}{\tau} \sqrt{\frac{2^k}{\tau}}\begin{bmatrix}\chi_{\left[0,\, \frac{\tau}{2^k}\right)} \\ 0 \\ 0\end{bmatrix}\otimes\begin{bmatrix} 0 \\ \chi_{\left[0,\, \frac{\tau}{2^k}\right)}  \\ 0\end{bmatrix}
\otimes\begin{bmatrix} 0 \\ 0 \\ \chi_{\left[0,\, \frac{\tau}{2^k}\right)}\end{bmatrix},
\end{split}\end{equation}
and $\mathcal{K}^k\subset \mathcal{F}^k$ is the linear span of all these vectors.

The interaction $M^k:\ \mathbb{H}\otimes \mathbb{K} \to \mathbb{H}\otimes \mathbb{K}$ 
between the initial system and one time step of the discrete noise 
is given by: 
\begin{equation}\label{definition M ex2}
M^k{ } = e^{\sqrt{\frac{\tau}{2^k}}\left(L_1\otimes \sigma^1_+ - L_1^*\otimes \sigma^1_-\right)}
           e^{\sqrt{\frac{\tau}{2^k}}\left(L_2\otimes \sigma^2_+ - L_2^*\otimes \sigma^2_-\right)}
          e^{\sqrt{\frac{\tau}{2^k}}\left(L_3\otimes \sigma^3_+ - L_3^*\otimes \sigma^3_-\right)}
          e^{\left( -i \frac{\tau}{2^k} H \otimes   I^{\otimes 3}\right)}
\end{equation}
Here the operator $L_1,L_2,L_3: \mathbb{H} \to \mathbb{H}$ couple the system 
to the three respective noise channels and the self-adjoint operator $H: \mathbb{H} \to \mathbb{H}$ is the Hamiltonian governing
the internal evolution of the initial system.

Note that if we choose a basis in $\mathbb{K}$ and we write $M^k$ as a matrix 
(with coefficients in $\mathbb{H}$) with respect to this 
basis,  then $\pi^k_2$ (defined in equation  \eqref{eqn pi2}) transforms this basis to a basis of 
$\mathcal{K}^k$ and $R^k$  (see equation \eqref{definition Rk}) has the same 
matrix with respect to this new basis as $M^k$ with respect to the old basis.

Now let $\alpha, \beta$ be elements of $\mathbb{C}^3$.
 If we define
\begin{equation*}
  \psi^k = 1 \oplus \alpha \chi_{\left[0,\, \frac{\tau}{2^k}\right)}, \qquad 
  \varphi^k = 1 \oplus \beta \chi_{\left[0,\, \frac{\tau}{2^k}\right)},
\end{equation*}
then $\psi^k, \varphi^k \in \mathcal{F}^k$ and due to Lemma \ref{lemma converg ex},
we have
\begin{equation*}
(\psi^{k})^{\otimes 2^k}\xrightarrow{k\to\infty}
	e(\alpha \chi_{[0,\tau)}),\qquad
(\varphi^{k})^{\otimes 2^k}\xrightarrow{k\to\infty}
	e(\beta \chi_{[0,\tau)}).
\end{equation*}

Let $u$  be an element of $\mathcal{H} =\mathbb{H}$. We take $u^k = u$ for all $k$, and we compute
\begin{gather*}
\frac{2^k}{\tau}(R^{k;\psi^k\varphi^k}-I)u^k = 
\frac{e^{i \frac{\tau}{2^k}H}}
{\sqrt{1+ \frac{\tau \sum_{i=1}^3|\alpha_i|^2 }{2^k}}\sqrt{1+ \frac{\tau \sum_{i=1}^3|\beta_i|^2}{2^k}}}
\Bigg(
\alpha_1^*\beta_1 \,  C_1  K_2 K_3 +
\alpha_2^*\beta_2  \, K_1 C_2 K_3 \\
+ \alpha_3^*\beta_3 \, K_1 K_2 C_3- \alpha_1^*\beta_2  L_1 J_1 L_2^\ast S_2 K_3 
-\alpha_1^*\beta_3\,    L_1 J_1 K_2 L_3^\ast S_3 \\
- \alpha_2^*\beta_1 \,  L^*_1 S_1 J_2 K_3 -  
\alpha_2^*\beta_3   \,  K_1 L_2 J_2 L^*_3 S_3 -
\alpha_3^*\beta_1 \,   L^*_1 S_1 K_2 L_3 J_3 \\
- \alpha_3^*\beta_2 \,    K_1 L_2^\ast S_2 L_3 J_3\\
+
\sqrt{\frac{2^k}{\tau}}\Bigg[
- \alpha_1^*  L_1 J_1K_2 K_3 - \alpha_2^* K_1  L_2 J_2K_3- \alpha_3^* K_1K_2 L_3 J_3 +\\
 \beta_1 L^*_1 S_1 K_2 K_3+ \beta_2 K_1 L_2^* S_2K_3+\beta_3 K_1K_2L_3^* S_3 \bigg]
 + \frac{2^k}{\tau}C_1 C_2 C_3 \Bigg)u  -  \frac{2^k}{\tau}u \\ 
\xrightarrow{k\to\infty}
\left(\sum_{i=1}^3 \left(\alpha_i^*\beta_i + \beta_i L_i^* - \alpha_i^* L_i - \frac{|\alpha_i|^2 + |\beta_i|^2  + L_i^*L_i}{2}\right) + iH\right)u,
\end{gather*}
\end{example}
where
\begin{eqnarray*}
    C_j= \cos(\sqrt{\frac{\tau}{2^k}L_jL_j^\ast} ), \quad K_j= \cos(\sqrt{\frac{\tau}{2^k}L_j^*L_j}),\\
     S_j= \frac{\sin(\sqrt{\frac{\tau}{2^k}L_jL_j^\ast})}{{\sqrt{L_jL_j^\ast}}}, \quad J_j= \frac{\sin(\sqrt{\frac{\tau}{2^k}L_j^*L_j})}{ {\sqrt{L_j^*L_j}}}.
\end{eqnarray*}
which we recognize, see Definition \ref{definition generator}, 
as the generator $\mathscr{L}^{\alpha \beta}$, acting on $u$, associated 
to the $SLH$-triple:
\begin{equation*}
  \left(I,\, \begin{bmatrix}L_1 \\ L_2\\ L_3\end{bmatrix},\, H\right).
\end{equation*}

Having checked all requirements for 
condition 1 of Theorem \ref{theorem BoutenvanHandel}, 
we can now conclude, by invoking the equivalence of condition 
1 and condition 3 of Theorem \ref{theorem BoutenvanHandel}, 
that (the embedding of) the repeated interaction of equation  \eqref{definition Rkl} with 
$M^k$ as in equation  \eqref{definition M ex2} converges strongly, uniformly on 
compact time intervals, to the solution $U_t$ of the 
following QSDE
\begin{equation*}\begin{split}
dU_t = \left\{\sum_{j=1}^3 \left(L_j dA^{j*}_{t} - L_j^* dA^j_t 
           -\frac{L_j^*L_j}{2}  dt\right) - iH dt \right\}U_t,
\end{split}\end{equation*}
with $U_0 = I$.

\begin{example}\label{example scattering}{\bf Generalized Beam Splitter}
In this example we consider optical elements that 
have no initial system, i.e.\ there is no absorption or 
emission of photons, only direct scattering between 
the different channels. The $SLH$-triple is just
\begin{eqnarray}
    \big( S, 0, 0 \big),
\end{eqnarray}
with $S \in U(m)$, a fixed unitary matrix of $m$ dimensions.

We take the one time step discrete noise space to be 
$\mathbb{K} = \big(\mathbb{C}^2\big)^{\otimes m}$.
That is, every channel is represented by a single qubit in a 
single time step. We can decompose $\mathbb{K}$
as a direct sum
\begin{equation*}
\mathbb{K} = \bigoplus_{i= 0}^m \mathbb{K}_i,
\end{equation*}
where $\mathbb{K}_0 = \mbox{span}\, \left\{|0^{\otimes m}\rangle \right\}$, the space spanned by the  vacuum vector, 
\begin{equation*}\begin{split}
&\mathbb{K}_1 =  \mbox{span}\,\Big\{ |1\otimes 0^{\otimes (m-1)}\rangle ,\ |0\otimes 1 \otimes 0^{\otimes (m-2)}\rangle , \ldots,\ |0^{\otimes (m-1)}\otimes 1\rangle \Big \},
\end{split}\end{equation*} 
the single-photon space, $\mathbb{K}_2$ is the two-photon space, and so on until $\mathbb{K}_m$ which 
is the space with $m$ photons. 
The interaction of the channels in 
one time step is given by 
\begin{equation}\label{definition MS}
M^k = I_0 \oplus S \oplus\bigoplus_{i =2}^m I_i,
\end{equation}
for some unitary operator $S: \mathbb{K}_1 \to \mathbb{K}_1$.
We now introduce the embedding  
$\pi^k:\ \mathbb{K} \to \mathcal{F}^k$ by
\begin{equation}\label{definition embedding}\begin{split}
&\pi^k\big(|0^{\otimes m}\rangle \big) := 1,\\ 
&\pi^k\big(|1\otimes 0^{\otimes (m-1)}\rangle \big) := \begin{bmatrix} \sqrt{\frac{2^k}{\tau}}\chi_{\left[0,\, \frac{\tau}{2^k}\right)} \\ 0 \\ \vdots \\ 0\end{bmatrix},\\
&\ \ \vdots\\
&\pi^k\big(|1^{\otimes m}\rangle \big) := \begin{bmatrix} \sqrt{\frac{2^k}{\tau}}\chi_{\left[0,\, \frac{\tau}{2^k}\right)}  \\ 0 \\ \vdots \\  0\end{bmatrix}\otimes\ldots\otimes
 \begin{bmatrix} 0  \\ \vdots \\ 0 \\ \sqrt{\frac{2^k}{\tau}}\chi_{\left[0,\, \frac{\tau}{2^k}\right)}\end{bmatrix}.
\end{split}\end{equation}
Note that with respect to the basis of equation  \eqref{definition embedding}
the embedded one-time-step interaction $R^k$ of equation \eqref{definition Rk}
is given by 
\begin{equation*}
R^k = I_0 \oplus S \oplus\bigoplus_{i =2}^m I_i,
\end{equation*}
Now, let $\alpha,\beta \in \mathbb{C}^m$ and let
\begin{equation*}
\psi^k := 1\oplus \begin{bmatrix} \alpha_1 \chi_{\left[0,\, \frac{\tau}{2^k}\right)} \\  \vdots \\  \alpha_m\chi_{\left[0,\, \frac{\tau}{2^k}\right)} \end{bmatrix},\ \ 
\varphi^k := 1\oplus \begin{bmatrix} \beta_1\chi_{\left[0,\, \frac{\tau}{2^k}\right)} \\  \vdots \\  \beta_m\chi_{\left[0,\, \frac{\tau}{2^k}\right)} \end{bmatrix}, 
\end{equation*}
then $\psi^k, \varphi^k \in \mathcal{F}^k$ and due to Lemma \ref{lemma converg ex}
we have
\begin{equation*}
(\psi^{k})^{\otimes 2^k}\xrightarrow{k\to\infty}
	e(\alpha \chi_{[0,\tau)}),\qquad
(\varphi^{k})^{\otimes 2^k}\xrightarrow{k\to\infty}
	e(\beta \chi_{[0,\tau)}).
\end{equation*}
We can now compute
\begin{equation*}\begin{split}
    \frac{2^k}{\tau}(R^{k;\psi^k\varphi^k}-I) = 
    & \frac{\frac{2^k}{\tau}+ \sum_{i,j=1}^m \alpha^*_iS_{ji}^*\beta_j     }  {\sqrt{(1+ \frac{\tau\sum_{i=1}^m|\alpha_i|^2}{2^k})(1+\frac{\tau\sum_{i=1}^m|\beta_i|^2}{2^k})}} -\frac{2^k}{\tau}\\
    &\xrightarrow{k\to\infty}\sum_{i,j=1}^m \alpha^*_jS_{ij}^*\beta_i - \frac{\sum_{i=1}^m \left(|\alpha_i|^2  + |\beta_i|^2\right)}{2},
\end{split}\end{equation*}
which we recognize, see Definition \ref{definition generator}, 
as the generator $\mathscr{L}^{\alpha \beta}$, associated 
to the triple $(S, 0, 0)$. Having checked all requirements for 
condition 1 of Theorem \ref{theorem BoutenvanHandel}, 
we can now conclude, by invoking the equivalence of condition 
1 and condition 3 of Theorem \ref{theorem BoutenvanHandel}, 
that (the embedding of) the repeated interaction of equation  \eqref{definition Rkl} with 
$M^k$ as in equation  \eqref{definition MS} converges strongly, uniformly on 
compact time intervals, to the solution $U_t$ of the 
following QSDE
\begin{equation}\label{equation scattering}
dU_t = \left\{\sum_{i,j=1}^m\big(S_{ij}-\delta_{ij}\big)d\Lambda^{ij}_t\right\}U_t, \qquad U_0= I.
\end{equation}
Note that since the interaction is non-trivial only 
on the single-photon space, we could also have 
used the result in \cite{AtP06} to prove the 
convergence. 

We will now look in more detail into the case where 
we have two channels: one in which the photons move
from right to left (labeled by $1$)  and one in which the 
photons move from left to right (labeled by $2$). 
Note that the $2\times 2$ matrix $S$ is unitary, i.e.\
in general it will be of the form:
\begin{equation*}
S = e^{i\gamma}\begin{bmatrix}e^{i\delta} \cos(\theta) & e^{-i\lambda}\sin(\theta) \\ -e^{i\lambda}\sin(\theta)& e^{-i\delta} \cos(\theta)\end{bmatrix},
\end{equation*}
for some real angles $\gamma, \delta,\lambda$ and $\theta$. 
We will set the global angle $\gamma$ equal to zero. We will now 
look at a perfect mirror: photons from channel $1$ arrive at the 
mirror and bounce off the mirror, continuing in channel $2$ while 
picking up a phase $\pi$. It is clear that this means $\theta =  \pi\slash 2$, 
$\delta$ is irrelevant, and $\lambda= 0$. Therefore, the scattering matrix
$S_{\text{mirror}}$ for a one-sided perfect mirror is given by
\begin{equation}\label{eqn S mirror}
S_{\mbox{mirror}} = \begin{bmatrix}0 & 1 \\ -1 & 0\end{bmatrix}.
\end{equation}
We can allow for some driving (from the left via channel $2$) via this mirror 
by taking
\begin{equation}\label{eqn S drivable mirror}
S_{\mbox{drivable mirror}} = \begin{bmatrix}\cos(\theta) & \sin(\theta) \\ -\sin(\theta) & \cos(\theta)\end{bmatrix}.
\end{equation}
We have chosen $\gamma =0$ for simplicity: in the 
next section we will build a cavity using two mirrors. We want the photons bouncing 
off the mirrors from the inside of the cavity to pick up a $\pi$ phase shift. The unitarity of $S$ then 
forces the driving photons that bounce off the mirrors from the other side not to pick up 
any phase shift.
\end{example}

\section{Examples: networks}\label{section examples networks}

\begin{example}\label{example Eschner}{\bf Light interference from a single ion and its mirror image \cite{ERSB01, BVS19}\ }
We consider the experiment described in \cite{ERSB01}. A single Barium ion $Ba^+$ sits 
in a Paul trap. The ion can be described as a two level atom: the ground state 
corresponds to the $S_{1/2}$-state and the excited state corresponds to the 
$P_{1/2}$-state. In the experiment \cite{ERSB01} the $D_{3/2}$ state and the 
transition $P_{1/2}\leftrightarrow D_{3/2}$ is used to reveal the 
population in the $P_{1/2}$ state. For simplicity, we will omit this part of the 
experiment from our simulated description. This means that our initial system 
$\mathcal{H} = \mathbb{H} = \mathbb{C}^2$ is a two level system, 
represented by a single qubit in a quantum computer. 

The Barium ion is driven from the side on its $S_{1/2}\leftrightarrow P_{1/2}$
transition and emits light (via focussing optics) towards a photo detector 
to the right of the ion and a mirror towards the left of the ion.
The light reflects from the mirror establishing a second light path towards 
the photo detector. The mirror position can be actuated  by a
piezoelectric stage, varying the difference in optical path lengths 
for the two light paths revealing an interference pattern at the photo detector \cite{ERSB01}.
Furthermore, the experiment showed the effect of the mirror position on 
the emission rate of the Barium ion: the $P_{1/2}$ state occupation probability 
was shown (via the $P_{1/2}\leftrightarrow D_{3/2}$ transition) to 
anti-correlate with the photon count rate at the detector \cite{ERSB01}. 

The experiment was simulated on a classical simulator of a digital 
quantum computer in \cite{BVS19}. There, the  model 
of example \ref{example 3 channels} was used to describe the 
Barium ion. Channel $1$ represents the photons moving from right 
to left, channel $2$ represents the photons moving from left to 
right and channel $3$ is perpendicular to channel $1$ and $2$ and 
represented the channel with which the Barium ion was driven by a 
laser. 

The time evolution of the Barium ion is given by 
\begin{equation*}\begin{split}
dU^{\text{Ba}}_t = \Bigg\{\sqrt{\kappa}&\sigma_-dA^{1*}_t -\sqrt{\kappa}\sigma_+dA^1_t +  
\sqrt{\kappa}\sigma_-dA^{2*}_t - \sqrt{\kappa}\sigma_+ dA^2_t  \\
&+ \sqrt{\kappa_3}\sigma_-dA^{3*}_t -\sqrt{\kappa_3}\sigma_+dA^3_t  
- \left(\kappa+\frac{\kappa_3}{2} - i\omega\right)\sigma_+\sigma_-dt \Bigg\}U^{\text{Ba}}_t, \qquad U^{\text{Ba}}_0 = I.
\end{split}\end{equation*}
That is, we are in the situation of Example \ref{example 3 channels} 
where the $SLH$-triple is given by
\begin{equation*}
 \left(I,\, \begin{bmatrix}L_1 \\ L_2\\ L_3\end{bmatrix},\, H\right) =  \left(I,\, \begin{bmatrix}\sqrt{\kappa}\sigma_- \\ \sqrt{\kappa}\sigma_-\\ \sqrt{\kappa_3}\sigma_-\end{bmatrix},\, \omega\sigma_+\sigma_-\right).
\end{equation*}
Here $\kappa$ is the \emph{decay rate} of the ion 
into channel $1$ and also the decay rate into channel $2$, $\kappa_3$ 
is the \emph{decay rate} into channel $3$ and $\omega$ is the \emph{transition 
frequency} of the transition  $S_{1/2}\leftrightarrow P_{1/2}$.

In the quantum computer the one-time step evolution of the Barium ion is given 
by the unitary of equation  \eqref{definition M ex2} on the space $\mathbb{H}_{\text{Ba}} \otimes \mathbb{K}_{\text{Ba}}$:
\begin{equation}\label{equation MkBa}
M^k_{\text{Ba} } = e^{\sqrt{\frac{\tau\kappa}{2^k}}\left(\sigma_-\otimes \sigma^1_+ -\sigma_+\otimes \sigma^1_-\right)}
           e^{\sqrt{\frac{\tau\kappa}{2^k}}\left(\sigma_-\otimes \sigma^2_+ - \sigma_+\otimes \sigma^2_-\right)}
          e^{\sqrt{\frac{\tau\kappa_3}{2^k}}\left(\sigma_-\otimes \sigma^3_+ - \sigma_+\otimes \sigma^3_-\right)}
          e^{\left( -i \frac{\tau\omega}{2^k} \sigma_+\sigma_- \otimes   I^{\otimes 3}\right)}
\end{equation}

The mirror couples two channels in the field. One channel 
consists of photons traveling from right to left (channel $4$) and the 
other consists of photons traveling from left to right (channel $5$) 
The time evolution of the mirror is given by
\begin{equation*}
dU^{\text{mirror}}_t = \left(d\Lambda^{45}_ t- d\Lambda^{54}_t -d\Lambda^{44}_t - d\Lambda^{55}_t\right)U^{\text{mirror}}_t, \qquad U^{\text{mirror}}_0 = I.
\end{equation*}
That is, we are in the situation of Example \ref{example scattering} with $SLH$-triple
\begin{equation*}
\left(\begin{bmatrix}0 & 1 \\ -1 & 0\end{bmatrix}, 0, 0\right).
\end{equation*}

In the quantum computer the one-time step evolution of the mirror is given by 
the unitary of equation  \eqref{definition MS} on the space $\mathbb{K}^{\text{mirror}} = \mathbb{C}^2\otimes \mathbb{C}^2$:
\begin{equation}\label{equation Mkmirror}
M^k_{\text{mirror}} = I_{\mathbb{K}^{\text{mirror}}_0} \oplus \begin{bmatrix} 0 & 1 \\ -1 & 0\end{bmatrix} \oplus I_{\mathbb{K}^{\text{mirror}}_2},
\end{equation}

We can describe the Barium ion and the mirror within one model 
(without making any connections between inputs and outputs yet) 
by the $SLH$-triple
\begin{equation}\label{equation Example4SLH}
\left(\begin{bmatrix}          1 & 0 & 0 & 0 & 0\\
                                       0 & 1 & 0 & 0 & 0\\
                                       0 & 0 & 1 & 0 & 0\\ 
                                       0 & 0 & 0 & 0 & 1 \\ 
                                       0 & 0 & 0 & -1 & 0\end{bmatrix},\ \begin{bmatrix}\sqrt{\kappa}\sigma_- \\ \sqrt{\kappa}\sigma_-\\ \sqrt{\kappa_3}\sigma_- \\ 0 \\0 \end{bmatrix},\, \omega\sigma_+\sigma_-\right).
\end{equation}
This means that the time evolution of the ion and mirror together (with no connections yet, i.e.\ $\mathcal{C} = \emptyset$) 
is given by 
\begin{equation}\label{equation Example4U}\begin{split}
dU_t = \Bigg\{d\Lambda^{45}_ t- d\Lambda^{54}_t& -d\Lambda^{44}_t - d\Lambda^{55}_t+ \sqrt{\kappa}\sigma_-dA^{1*}_t -\sqrt{\kappa}\sigma_+dA^1_t +  
\sqrt{\kappa}\sigma_-dA^{2*}_t - \sqrt{\kappa}\sigma_+ dA^2_t  \\
&+ \sqrt{\kappa_3}\sigma_-dA^{3*}_t -\sqrt{\kappa_3}\sigma_+dA^3_t  
- \left(\kappa+\frac{\kappa_3}{2} - i\omega\right)\sigma_+\sigma_-dt \Bigg\}U_t, \qquad U_0 = I.
\end{split}\end{equation}
We now introduce the following set of connections (Note that Assumption \ref{assumption multiples of tau} is satisfied):
\begin{equation}\label{equation Example4ConnectionSet}
\mathcal{C}=  \Big\{(4,\, 1,\, n\tau),\  (2,\, 5,\, n\tau) \Big\}, \qquad \qquad \mbox{for some}\ n \in \mathbb{N}\backslash \{0\}.
\end{equation} 
This means that the photons going from right to left, after having interacted with the Barium ion, are being fed into the 
input channel $4$ of the mirror. That is, instead of disappearing to infinity on the left, the photons move towards the mirror where 
they reflect into output channel $5$ picking up a $\pi$-phase
flip. Instead of moving straight away to infinity on the right, they are then 
subsequently fed into channel $2$. That is, they continue moving from left to right, but they are now moving 
towards the Barium ion where they will interact with the ion once more before disappearing to infinity on the right.

The time evolution $\hat{\mathcal{U}}_t$ of the network is defined in equation \eqref{definition CalUhat} 
with $\mathcal{U}_t$ given by Definition \ref{definition time evolution network}. Note that in 
Definition \ref{definition time evolution network} the cocycle $U_t$ is given by equation \eqref{equation Example4U} 
or, equivalently, by the $SLH$-triple of equation \eqref{equation Example4SLH}. The set of 
connections $\mathcal{C}$ given in equation \eqref{equation Example4ConnectionSet} determines 
the braid map $B_\sigma$ in Definition \ref{definition time evolution network}. 

Channels $1,2,4$ and $5$ are 
at $t=0$ in the vacuum state $\Phi: = e(0)$, see equation \eqref{def exponential vector}. At $t=0$, channel 
$3$ is in a coherent state $\psi(f)$, where 
\begin{equation}\label{equation coherent state}
\psi(f) = \exp(-\tfrac{1}{2}\|f\|^2)e(f),\qquad
f(t) = A\exp(i\omega_{l} t)\chi_{[0,T]}(t), 
\end{equation}
representing a laser driving the system at 
frequency $\omega_l$ with amplitude $A$ over a time interval $[0,T]$. Note that 
the laser is resonant when $\omega_l = \omega$.

Let $\tau$ be the interval from Assumption \ref{assumption multiples of tau}. 
 For simplicity let us assume $T$ is an integer multiple of $\tau$. We then have
\begin{equation*}
\mathsf{N}_k = \frac{2^k(T+\xi_{\text{max}})}{\tau}.
\end{equation*} 
In the quantum computer, one qubit is used for the Barium ion $\mathbb{H} = \mathbb{C}^2$. 
One time step of the noise consists of $5$ qubits: $\mathbb{K} = (\mathbb{C}^2)^{\otimes 5}$.
There are $\mathsf{N}_k$ time steps, i.e.\ in the quantum computer we reserve $1+5*\mathsf{N}_k$ 
qubits to obtain the computational space $\mathbb{H}\otimes \mathbb{K}^{\mathsf{N}_k}$.
We define 
\begin{equation*}\begin{split}
& \pi_1^k:\ \mathbb{H} \to \mathcal{H} = \mathbb{H}: \ v \mapsto v,\\
&\pi_2^k:\ \mathbb{K} \to \mathcal{F}^k,\qquad \pi^k_2 := \pi^k_{2,\, \text{Ex.\ref{example 3 channels}}}\otimes \pi^k_{2,\, \text{Ex.\ref{example scattering}}}, 
\end{split}\end{equation*}  
where $\pi^k_{2,\, \text{Ex.\ref{example 3 channels}}}$ is the $\pi_2^k$ of Example \ref{example 3 channels} which acts 
non-trivially on the first three components of the tensor product $\mathbb{K}$ and  $\pi^k_{2,\, \text{Ex.\ref{example scattering}}}$ is the
$\pi^k$ of Example \ref{example scattering} which acts non-trivially on the last two components of the tensor product $\mathbb{K}$.

At the start of the computation, all qubits in $\mathbb{H}\otimes\mathbb{K}^{\mathsf{N}_k}$ 
are in the down state $|0>$. Using single qubit gates we can bring the qubit 
that represents the Barium ion in any state that we desire: the initial state of the 
Barium ion can be chosen freely. We leave all the qubits in channels $1,2,4$ and 
$5$ in the $|0>$ state. Let $f^1,f^2,\ldots \in L^2(\mathbb{R})$ be a sequence of simple functions 
that approximate $f$ (see equation \eqref{equation coherent state})
in the sense of Lemma \ref{lemma converg ex}, i.e.\ 
there exist $\alpha^k_1,\alpha^k_2,\dots,\alpha^k_{\mathsf{N}_k} \in \mathbb{C}$ such that
\begin{equation*}
f^k =  \sum_{l=1}^{\frac{2^kT}{\tau}} \alpha^k_l \chi_{\left[\frac{(l-1)\tau}{ 2^k},\, \frac{l\,\tau}{2^k}\right)},\qquad \left\| f- f^k\right\|_2  \xrightarrow{k\to\infty} 0.
\end{equation*}
Using single qubit gates we bring the $l^{th}$ qubit, for $1\le l\le \frac{2^kT}{\tau}$, in channel $3$ in the state 
\begin{equation*}
\psi^k_l =  \frac{|0> + \alpha^k_l\sqrt{\frac{\tau}{2^k}}|1>}{\sqrt{1+|\alpha^k_l|^2\frac{\tau}{2^k} }}.
\end{equation*}
All the other qubits in channel $3$ are left in the $|0>$ state, i.e. $\alpha^k_l = 0$ for $\frac{2^kT}{\tau} < l \le \mathsf{N}_k$. Lemma \ref{lemma converg ex} yields
\begin{equation*}\begin{split}
&(\pi^k_2)^{\otimes \mathsf{N}_k}\left(
\bigotimes_{l=1}^{\mathsf{N_k}}
|00>\otimes\left(\frac{|0> + \alpha^k_l\sqrt{\frac{\tau}{2^k}}|1>}{\sqrt{1+|\alpha^k_l|^2\frac{\tau}{2^k} }}\right)\otimes|00>
\right) = { }
\bigotimes_{l=1}^{\mathsf{N_k}}
\frac{1 \oplus
 \alpha^k_l{\tiny\begin{bmatrix}
0\\
0\\
\chi_{\left[\frac{(l-1)\tau}{2^k},\, \frac{l\tau}{2^k}\right)}\\
0\\
0
\end{bmatrix}}
}{\sqrt{1+|\alpha^k_l|^2\frac{\tau}{2^k}   }}\\
& \xrightarrow{k\to\infty} \exp(-\tfrac{1}{2}\|f\|^2)e\left(
\begin{bmatrix}
0 \\
0 \\
f \\
0\\
0
\end{bmatrix}\right)
= 1\otimes 1 \otimes \psi(f)\otimes 1\otimes 1,
\end{split}\end{equation*}
 where $\psi(f)$ is the coherent state from equation \eqref{equation coherent state}.

On the one time step quantum computer space $\mathbb{H}\otimes\mathbb{K}$
we define
\begin{equation*}
M^k := M^k_{\text{Ba}}\otimes M^k_{\text{mirror}},
\end{equation*}
where $M^k_{\text{Ba}}$ is given by equation \eqref{equation MkBa} and 
$M^k_{\text{mirror}}$ is given by equation \eqref{equation Mkmirror}.
$M^k_{\text{Ba}}$ acts on $\mathbb{H}$ and the first three copies of $\mathbb{C}^2$
in $\mathbb{K}$. $M^k_{\text{mirror}}$ acts on the last two copies of $\mathbb{C}^2$ 
in $\mathbb{K}$. Using Definition \ref{definition cal M} we obtain a 
unitary process $\mathcal{M}^k(l)$ for $0 \le l \le \lfloor\frac{2^k}{\tau}\rfloor$
on the quantum computer space $\mathbb{H}\otimes \mathbb{K}^{\otimes \mathsf{N}_k}$. 
The process $\mathcal{M}^k(l)$ is the digital twin of the network dynamics $\mathcal{U}_t$ given by 
equation \eqref{equation Example4U} and Definition \ref{definition time evolution network}. 
We will now show that $\mathcal{M}^k(l)$ indeed converges to $\mathcal{U}_t$.

Note that equation \eqref{definition Rk} defines $R_k$ on $\mathcal{H}\otimes\mathcal{K}^k \subset \mathcal{H} \otimes\mathcal{F}^k$.
Next, equation \eqref{definition Rkt} defines $R^k_t$ on $\mathcal{H}\otimes \mathcal{F}$ and then Definition  \ref{definition discrete evolution network}
defines the discrete network time evolution $\mathcal{R}_t^k$ on $\mathcal{H}\otimes \mathcal{F}$.
Lemma \ref{lemma CalRk} then summarizes what we have thus far: for all $0 \le t \le T$ and all 
$\psi \in \mathcal{H}\otimes \big( {\mathcal{K}^k}\big)^{\otimes \mathsf{N}_k}\subset \mathcal{H}\otimes \mathcal{F}_{[0,\, \mathsf{N}_k\tau 2^{-k}]}$
\begin{equation*}
    \mathcal{R}_t^k\psi = 
\bigg( \pi_1^k\otimes \big({\pi^k_2}\big)^{\otimes \mathsf{N}_k} \bigg)
\mathcal{M}^k\left(\left\lfloor \frac{t2^k}{\tau} \right\rfloor\right)
\bigg( \pi_1^{k}\otimes \big( {\pi^{k}_2}\big) ^{\otimes \mathsf{N}_k}\bigg)^\ast \, \psi.
\end{equation*}

Having proved convergence for the components of the network in 
Example \ref{example 3 channels} (Barium ion) and Example \ref{example scattering}
(mirror) using \cite[Thm 1]{BvH08},  
we conclude from Corollary \ref{corollary} that 
\begin{equation*}
	\lim_{k\to\infty}\sup_{0\le t\le T}
	\left\|\bigg( \pi_1^k\otimes \big({\pi^k_2}\big)^{\otimes \mathsf{N}_k} \bigg)
\mathcal{M}^k\left(\left\lfloor \frac{t2^k}{\tau} \right\rfloor\right)
\bigg( \pi_1^{k}\otimes \big( {\pi^{k}_2}\big) ^{\otimes \mathsf{N}_k}\bigg)^\ast \, \psi-\mathcal{U}_t\psi\right\|=0, \qquad \forall \psi \in \mathcal{H}\otimes\mathcal{F},
\end{equation*}
which indeed shows that the discrete process $\mathcal{M}^k(l),\ 0\le l \le \lfloor \frac{2^kT}{\tau}\rfloor$ running on the 
quantum computer converges to the network 
time evolution $\mathcal{U}_t,\ 0\le t\le T$ after embedding it in the Fock space.

We have not optimized our description with respect to 
qubit use. Every qubit in our description 
that is never acted on by $M^k$ at some moment in the 
computation can be eliminated from the description. This 
means that all channels that have their input connected to 
an output can be eliminated up to a length corresponding 
to the length of the connection in question, see Remark \ref{remark qubits needed}. 
Furthermore, 
due to re-initialization and recycling of qubits, 
the number of qubits used in the simulation on a 
classical simulator in \cite{BVS19} was
drastically reduced to only $6$ qubits. One 
for the Barium ion, one for the laser driving from 
the side, and four for the internal line. Nevertheless, 
the simulation reproduced the interference pattern 
and modified emission rate of the ion 
seen in the experiment \cite{ERSB01, BVS19}.
\end{example}

\begin{example}\label{example cavity QED}{\bf Cavity QED\ }
In this example we remain in the situation of Example \ref{example Eschner} 
but add an extra mirror to the description and we change the Barium 
ion for a generic two-level atom. The extra mirror is placed 
to the right of the two-level atom. The mirror interacts with 
two channels in the field representing photons moving from 
left to right (channel $6$) and photons moving from right 
to left (channel $7$). We allow the mirror to be driven, i.e.\ it 
is described by (see equations \eqref{equation scattering} and \eqref{eqn S drivable mirror}):
\begin{equation*}
dU^{\text{driv mir}}_t = \left(\sin(\theta)\Big(d\Lambda^{76}_ t-  d\Lambda^{67}_t\Big) + (\cos(\theta)-1)\Big( d\Lambda^{66}_t + d\Lambda^{77}_t\Big)\right)U^{\text{driv mir}}_t, \qquad  U^{\text{driv mir}}_0 = I.
\end{equation*}
This means that the entire system of two-level atom and 
mirrors together, with no connections implemented yet, is given by: 
\begin{equation}\label{equation Example5U}\begin{split}
dU_t = \Bigg\{d\Lambda^{45}_ t- d\Lambda^{54}_t& -d\Lambda^{44}_t - d\Lambda^{55}_t + 
\sin(\theta)\Big(d\Lambda^{76}_ t-  d\Lambda^{67}_t\Big) + (\cos(\theta)-1)\Big( d\Lambda^{66}_t + d\Lambda^{77}_t\Big)\\
&+ \sqrt{\kappa}\sigma_-dA^{1*}_t -\sqrt{\kappa}\sigma_+dA^1_t +  
\sqrt{\kappa}\sigma_-dA^{2*}_t - \sqrt{\kappa}\sigma_+ dA^2_t  \\
&+ \sqrt{\kappa_3}\sigma_-dA^{3*}_t -\sqrt{\kappa_3}\sigma_+dA^3_t  
- \left(\kappa+\frac{\kappa_3}{2} - i\omega\right)\sigma_+\sigma_-dt \Bigg\}U_t, \qquad U_0 = I.
\end{split}\end{equation}
Or, equivalently, it is given by the $SLH$-triple:
\begin{equation}\label{equation Example5SLH}
\left(\begin{bmatrix}          1 & 0 & 0 & 0 & 0 & 0 & 0\\
                                       0 & 1 & 0 & 0 & 0 & 0 & 0\\
                                       0 & 0 & 1 & 0 & 0 & 0 & 0\\ 
                                       0 & 0 & 0 & 0 & 1 & 0 & 0\\ 
                                       0 & 0 & 0 & -1 & 0 & 0 & 0 \\
                                       0 & 0 & 0 & 0 & 0 & \cos(\theta) & -\sin(\theta)\\ 
                                       0 & 0 & 0 & 0 & 0 & \sin(\theta) & \cos(\theta)\end{bmatrix},\ \begin{bmatrix}\sqrt{\kappa}\sigma_- \\ \sqrt{\kappa}\sigma_-\\ \sqrt{\kappa_3}\sigma_- \\ 0 \\0 \\0 \\0\end{bmatrix},\, \omega\sigma_+\sigma_-\right).
\end{equation}
Let $n,n_1$ and $n_2$ be elements of $\mathbb{N}\backslash \{0\}$ such that $n_1 + n_2 =n$. 
We let the connection set be given by
\begin{equation}\label{equation Example5ConnectionSet}
\mathcal{C}=  \Big\{(4,\, 1,\, n_1\tau),\  (2,\, 5,\, n_1\tau),\ (7,\, 2,\, n_2\tau),\ (1,\, 6,\, n_2\tau)  \Big\},
\end{equation}
which satisfies Assumption \ref{assumption multiples of tau}. 
The $SLH$-network given by the $SLH$-triple of equation \eqref{equation Example5SLH} and 
connection set of equation \eqref{equation Example5ConnectionSet} describes a cavity where 
light bounces between two mirrors. Situated inside the cavity there is a two-level atom that 
interacts with the light in the cavity. The position of the two-level atom in the cavity is variable 
and determined by how $n_1$ and $n_t$ add up to the fixed number $n$. 
The right mirror can be driven by e.g.\ a laser on channel $6$ and 
the two-level atom can be driven from the side (channel $3$) by e.g.\ a laser. 
The network that we have created, a two level atom in a cavity, is the archetypical 
example studied in cavity QED.

Let $\tau$ be the interval from Assumption \ref{assumption multiples of tau}. 
 For simplicity let us assume $T$ is an integer multiple of $\tau$. We then have
\begin{equation*}
\mathsf{N}_k = \frac{2^k(T+\xi_{\text{max}})}{\tau}.
\end{equation*} 
In the quantum computer, one qubit is used for the two-level atom $\mathbb{H} = \mathbb{C}^2$. 
One time step of the noise consists of $7$ qubits: $\mathbb{K} = (\mathbb{C}^2)^{\otimes 7}$.
There are $\mathsf{N}_k$ time steps, i.e.\ in the quantum computer we reserve $1+7*\mathsf{N}_k$ 
qubits to obtain the computational space $\mathbb{H}\otimes \mathbb{K}^{\mathsf{N}_k}$.
We define 
\begin{equation*}\begin{split}
& \pi_1^k:\ \mathbb{H} \to \mathcal{H} = \mathbb{H}: \ v \mapsto v,\\
&\pi_2^k:\ \mathbb{K} \to \mathcal{F}^k,\qquad \pi^k_2 := \pi^k_{2,\, \text{Ex.\ref{example 3 channels}}}
\otimes \pi^k_{2,\, \text{Ex.\ref{example scattering}}}\otimes \pi^k_{2,\, \text{Ex.\ref{example scattering}}}, 
\end{split}\end{equation*}  
where $\pi^k_{2,\, \text{Ex.\ref{example 3 channels}}}$ is the $\pi_2^k$ of Example \ref{example 3 channels} which acts 
non-trivially on the first three components of the tensor product $\mathbb{K}$ and  $\pi^k_{2,\, \text{Ex.\ref{example scattering}}}$ is the
$\pi^k$ of Example \ref{example scattering} which acts non-trivially on the fourth and fifth component of the tensor product $\mathbb{K}$ 
and the sixth and seventh component of the tensor product.

At the start of the computation, all qubits in $\mathbb{H}\otimes\mathbb{K}^{\mathsf{N}_k}$ 
are in the down state $|0>$. Using single qubit gates we can bring the qubit 
that represents the two level atom in any state that we desire: the initial state of the 
two level system can be chosen freely. We leave all the qubits in channels $1,2,3,4,5$ and 
$7$ in the $|0>$ state (no driving from the side, only on the drivable mirror). 
Let $f^1,f^2,\ldots \in L^2(\mathbb{R})$ be a sequence of simple functions 
that approximate $f$ (see equation \eqref{equation coherent state})
in the sense of Lemma \ref{lemma converg ex}, i.e.\ 
there exist $\alpha^k_1,\alpha^k_2,\dots,\alpha^k_{\mathsf{N}_k} \in \mathbb{C}$ such that
\begin{equation*}
f^k =  \sum_{l=1}^{\frac{2^kT}{\tau}} \alpha^k_l \chi_{\left[\frac{(l-1)\tau}{ 2^k},\, \frac{l\,\tau}{2^k}\right)},\qquad \left\| f- f^k\right\|_2  \xrightarrow{k\to\infty} 0.
\end{equation*}
Using single qubit gates we bring the $l^{th}$ qubit, for $1\le l\le \frac{2^kT}{\tau}$, in channel $6$ in the state 
\begin{equation*}
\psi^k_l =  \frac{|0> + \alpha^k_l\sqrt{\frac{\tau}{2^k}}|1>}{\sqrt{1+|\alpha^k_l|^2\frac{\tau}{2^k} }}.
\end{equation*}
All the other qubits in channel $6$ are left in the $|0>$ state, i.e. $\alpha^k_l = 0$ for $\frac{2^kT}{\tau} < l \le \mathsf{N}_k$. Lemma \ref{lemma converg ex} yields
\begin{equation*}\begin{split}
&(\pi^k_2)^{\otimes \mathsf{N}_k}\left(
\bigotimes_{l=1}^{\mathsf{N_k}}
|00000>\otimes\left(\frac{|0> + \alpha^k_l\sqrt{\frac{\tau}{2^k}}|1>}{\sqrt{1+|\alpha^k_l|^2\frac{\tau}{2^k} }}\right)\otimes|0>
\right) = { }
\bigotimes_{l=1}^{\mathsf{N_k}}
\frac{1 \oplus
 \alpha^k_l{\tiny\begin{bmatrix}
0\\
0\\
0\\
0\\
0\\
\chi_{\left[\frac{(l-1)\tau}{2^k},\, \frac{l\tau}{2^k}\right)}\\
0
\end{bmatrix}}
}{\sqrt{1+|\alpha^k_l|^2\frac{\tau}{2^k}   }}\\
& \xrightarrow{k\to\infty} \exp(-\tfrac{1}{2}\|f\|^2)e\left(
\begin{bmatrix}
0 \\
0 \\
0 \\
0 \\
0 \\
f \\
0
\end{bmatrix}\right)
= 1\otimes 1 \otimes 1\otimes 1\otimes 1\otimes \psi(f)\otimes 1,
\end{split}\end{equation*}
 where $\psi(f)$ is the coherent state from equation \eqref{equation coherent state}.

On the one time step quantum computer space $\mathbb{H}\otimes\mathbb{K} = 
\mathbb{H}\otimes\mathbb{K}^{\text{atom}}\otimes\mathbb{K}^{\text{mirror}}\otimes \mathbb{K}^{\text{driv mir}}$, 
where $\mathbb{K}^{\text{atom}} = \mathbb{C}^2\otimes\mathbb{C}^2\otimes  \mathbb{C}^2,\  
\mathbb{K}^{\text{mirror}} = \mathbb{C}^2\otimes\mathbb{C}^2$, and $\mathbb{K}^{\text{driv mir}}= \mathbb{C}^2\otimes\mathbb{C}^2$,
we define
\begin{equation*}
M^k := M^k_{\text{atom}}\otimes M^k_{\text{mirror}} \otimes M^k_{\text{drivable mirror}},
\end{equation*}
where $M^k_{\text{atom}}$ is given by equation \eqref{definition M ex2} and 
\begin{equation*}\begin{split}
&M^k_{\text{mirror}} = I_{\mathbb{K}^{\text{mirror}}_0} \oplus \begin{bmatrix} 0 & 1 \\ -1 & 0\end{bmatrix} \oplus I_{\mathbb{K}^{\text{mirror}}_2},\\
&M^k_{\text{drivable mirror}} = I_{\mathbb{K}^{\text{driv mir}}_0} \oplus 
\begin{bmatrix}\cos(\theta) & \sin(\theta) \\ -\sin(\theta) & \cos(\theta)\end{bmatrix} \oplus I_{\mathbb{K}^{\text{driv mir}}_2},
\end{split}\end{equation*}
$M^k_{\text{atom}}$ acts on $\mathbb{H}$ and the first three copies of $\mathbb{C}^2$
in $\mathbb{K}$. $M^k_{\text{mirror}}$ acts on the fourth and fifth copies of $\mathbb{C}^2$ 
in $\mathbb{K}$. $M^k_{\text{drivable mirror}}$ acts on the sixth and seventh copies of $\mathbb{C}^2$ 
in $\mathbb{K}$.
Using Definition \ref{definition cal M} we obtain a 
unitary process $\mathcal{M}^k(l)$ for $0 \le l \le \lfloor\frac{2^k}{\tau}\rfloor$
on the quantum computer space $\mathbb{H}\otimes \mathbb{K}^{\otimes \mathsf{N}_k}$. 
The process $\mathcal{M}^k(l)$ is the digital twin of the network dynamics $\mathcal{U}_t$ given by 
equation \eqref{equation Example5U} and Definition \ref{definition time evolution network}. 
We will now show that $\mathcal{M}^k(l)$ indeed converges to $\mathcal{U}_t$.

Note that equation \eqref{definition Rk} defines $R_k$ on $\mathcal{H}\otimes\mathcal{K}^k \subset \mathcal{H} \otimes\mathcal{F}^k$.
Next, equation \eqref{definition Rkt} defines $R^k_t$ on $\mathcal{H}\otimes \mathcal{F}$ and then Definition  \ref{definition discrete evolution network}
defines the discrete network time evolution $\mathcal{R}_t^k$ on $\mathcal{H}\otimes \mathcal{F}$.
Lemma \ref{lemma CalRk} then summarizes what we have thus far: for all $0 \le t \le T$ and all 
$\psi \in \mathcal{H}\otimes \big( {\mathcal{K}^k}\big)^{\otimes \mathsf{N}_k}\subset \mathcal{H}\otimes \mathcal{F}_{[0,\, \mathsf{N}_k\tau 2^{-k}]}$
\begin{equation*}
    \mathcal{R}_t^k\psi = 
\bigg( \pi_1^k\otimes \big({\pi^k_2}\big)^{\otimes \mathsf{N}_k} \bigg)
\mathcal{M}^k\left(\left\lfloor \frac{t2^k}{\tau} \right\rfloor\right)
\bigg( \pi_1^{k}\otimes \big( {\pi^{k}_2}\big) ^{\otimes \mathsf{N}_k}\bigg)^\ast \, \psi.
\end{equation*}

Having proved convergence for the components of the network in 
Example \ref{example 3 channels} (two-level atom) and Example \ref{example scattering}
(mirror and drivable mirror) using \cite[Thm 1]{BvH08},  
we conclude from Corollary \ref{corollary} that 
\begin{equation*}
	\lim_{k\to\infty}\sup_{0\le t\le T}
	\left\|\bigg( \pi_1^k\otimes \big({\pi^k_2}\big)^{\otimes \mathsf{N}_k} \bigg)
\mathcal{M}^k\left(\left\lfloor \frac{t2^k}{\tau} \right\rfloor\right)
\bigg( \pi_1^{k}\otimes \big( {\pi^{k}_2}\big) ^{\otimes \mathsf{N}_k}\bigg)^\ast \, \psi-\mathcal{U}_t\psi\right\|=0, \qquad \forall \psi \in \mathcal{H}\otimes\mathcal{F},
\end{equation*}
which indeed shows that the discrete process $\mathcal{M}^k(l),\ 0\le l \le \lfloor \frac{2^kT}{\tau}\rfloor$ running on the 
quantum computer converges to the network 
time evolution $\mathcal{U}_t,\ 0\le t\le T$ after embedding it in the Fock space.

We could again use similar methods as in Remark \ref{remark qubits needed} \cite{BVS19} to 
reduce the number of qubits needed: e.g. one for the atom, one for the driving on the drivable 
mirror (measuring and re-initializing the qubit in each time step)
and e.g.\ $16$ or $32$ qubits for the loop representing the cavity. This would be feasible on a classical 
simulator of a digital quantum computer.  One could e.g.\ investigate exciting the lowest resonance modes 
of the cavity by changing the frequency of the driving on the mirror, the dependency of the 
atom field coupling $g$ on the position of the atom in the cavity or the Lorentzian response around 
a cavity resonance. It would be even more interesting to see this on real quantum 
hardware.
\end{example}

\section{Proof of the main result: Theorem \ref{theorem main result}}\label{section proofs}

The proof of our main result, i.e.\ Theorem \ref{theorem main result}, 
depends heavily on Theorem \ref{theorem BoutenvanHandel} \cite[Thm 1]{BvH08}, 
which itself depends heavily on a result due to 
Thomas Kurtz \cite[Thm 2.13]{Kur69}, see also \cite[Ch.\,1, Thm 6.5]{KuE86}.
In the proof of Theorem \ref{theorem main result} we will also 
need to rely directly on the result \cite[Ch.\,1, Thm 6.5]{KuE86} (see 
also \cite{Gol76}) which we reproduce here in 
a form convenient to us, i.e.\ specialized to our particular 
use case.

\begin{theorem}{\bf \cite[Ch.\,1, Thm 6.5]{KuE86}}
\label{theorem Kurtz}
Let $\mathcal{H}$ be a fixed Hilbert space.  For $k\in\mathbb{N}$, let 
$T^k$ be a linear contraction on $\mathcal{H}$ and let $\varepsilon_k$ be a positive 
number such that $\lim_{k\to \infty} \varepsilon_k = 0$. Let $T_t$ be a strongly 
continuous contraction semigroup on $\mathcal{H}$ with generator
$\mathscr{L}$.  Let $\mathcal{D}$ be a core for $\mathscr{L}$.  The 
following conditions are equivalent:
\begin{enumerate}
\item For every $u\in\mathcal{D}$, there exist $u^k\in\mathcal{H}$ such
that
$$
	u^k\xrightarrow{k\to\infty}u,\qquad
	\varepsilon_k^{-1}(T^k-I)u^k\xrightarrow{k\to\infty}\mathscr{L}u.
$$
\item For every $\psi\in\mathcal{H}$ and $t<\infty$
$$
	\lim_{k\to\infty}
	\|(T^k)^{\lfloor t\slash\varepsilon_k\rfloor}\psi-T_t\psi\|=0.
$$
\item For every $\psi\in\mathcal{H}$ and $t<\infty$
$$
	\lim_{k\to\infty}\sup_{s\le t}
	\|(T^k)^{\lfloor s\slash\varepsilon_k\rfloor}\psi-T_s\psi\|=0.
$$
\end{enumerate}
\end{theorem}

{\bf Proof of Theorem \ref{theorem main result}}\\
$1\Rightarrow 2:$ Let $\varphi,\varphi_k \in \mathcal{H}\otimes\mathcal{F}$ be such that 
$\varphi_k \xrightarrow{k\to\infty} \varphi$. Note that then
\begin{equation}\begin{split}\label{eq convergence Rtkphik}
\|R_t^{k}\varphi_k -U_t\varphi\|{ } &= \|R_t^k\varphi_k  -U_t\varphi  + R_t^k\varphi -R_t^k\varphi\| \\
                                    &\le \|R_t^k(\varphi_k-\varphi)\| + \|(R_t^k-U_t)\varphi\| =  \|\varphi_k-\varphi\| + \|(R_t^k-U_t)\varphi\| \xrightarrow{k\to\infty} 0,
\end{split}\end{equation}
where we have used in the last step that condition 1 combined with Theorem \ref{theorem BoutenvanHandel}
implies that $\|(R_t^k-U_t)\varphi\| \xrightarrow{k\to\infty} 0$. 

Let $\psi \in \mathcal{H}\otimes\mathcal{F}$. If we define $\varphi_k = \Theta_\sigma B_\sigma R_\sigma^k\psi$ and 
$\varphi = \Theta_\sigma B_\sigma U_\sigma\psi$, then 
due to the unitarity of $\Theta_\sigma$ and $B_\sigma$ and due to condition 1 combined with Theorem \ref{theorem BoutenvanHandel}, 
we find $\varphi_k  \xrightarrow{k\to\infty} \varphi$. equation  \eqref{eq convergence Rtkphik} then yields 
that $R_\sigma^k\varphi_k \xrightarrow{k\to\infty} U_\sigma\phi$. We can now redefine 
$\varphi_k =  \Theta_\sigma B_\sigma R_\sigma^k\Theta_\sigma B_\sigma R_\sigma^k\psi$ and $\varphi = \Theta_\sigma B_\sigma U_\sigma\Theta_\sigma B_\sigma U_\sigma\psi$
and we immediately have $\varphi_k \xrightarrow{k\to\infty} \varphi$. Using equation  \eqref{eq convergence Rtkphik} as before, 
we can then conclude $R_\sigma^k\varphi_k \xrightarrow{k\to\infty} U_\sigma\varphi$. By iterating this 
procedure sufficiently many times, we find that 
\begin{equation*}
	\lim_{k\to\infty}
	\|\mathcal{R}_t^{k}\psi-\mathcal{U}_t\psi\|=0,\qquad \forall \psi \in\mathcal{H}\otimes\mathcal{F}.
\end{equation*}

$2\Rightarrow 1:$
Note that it follows from the definition of $R_t^k$ in equation  \eqref{definition Rkt} and the 
co-cycle property of $U_t$ equation  \eqref{cocycle Ut} that
\begin{equation*}\begin{split}
&R_t^k =  \Theta_{\lfloor \frac{t}{\sigma }\rfloor \sigma }^* R^k_{t-\lfloor\frac{t}{\sigma }\rfloor \sigma }\Theta_{\lfloor \frac{t}{\sigma }\rfloor \sigma }
                     \prod_{i=1}^{\lfloor \frac{t}{\sigma }\rfloor} \Theta_{(i-1)\sigma }^*R^k_\sigma  \Theta_{(i-1)\sigma },\\
&U_t = \Theta_{\lfloor \frac{t}{\sigma }\rfloor \sigma }^* U_{t-\lfloor\frac{t}{\sigma }\rfloor \sigma }\Theta_{\lfloor \frac{t}{\sigma }\rfloor \sigma }
                     \prod_{i=1}^{\lfloor \frac{t}{\sigma }\rfloor} \Theta_{(i-1)\sigma }^*U_\sigma  \Theta_{(i-1)\sigma },\\
\end{split}\end{equation*}
for all $0< \sigma \le \xi_{\text{min}}$ such that $\sigma $ is an integral multiple of $\tau 2^{-k}$. We will take 
$\sigma  = \tau$ in the following. Since for all $0\le s\le \tau :$ $R_s^k = B_s^* \mathcal{R}^k_s$ due to Definition 
\ref{definition discrete evolution network} and $U_s = B_s^* \mathcal{U}_s$ due to 
Definition \ref{definition time evolution network}, we find
\begin{equation}\begin{split}\label{eq expressions Rtk U}
&R_t^k =  \Theta_{\lfloor \frac{t}{\tau }\rfloor \tau }^* B^*_{t-\lfloor\frac{t}{\tau }\rfloor\tau }\mathcal{R}^k_{t-\lfloor\frac{t}{\tau }\rfloor \tau }\Theta_{\lfloor \frac{t}{\tau }\rfloor \tau }
                     \prod_{i=1}^{\lfloor \frac{t}{\tau }\rfloor} \Theta_{(i-1)\tau }^*B_{\tau }^*\mathcal{R}^k_{\tau } \Theta_{(i-1)\tau },\\
&U_t = \Theta_{\lfloor \frac{t}{\tau }\rfloor \tau }^* B_{t-\lfloor\frac{t}{\tau }\rfloor \tau }^*\mathcal{U}_{t-\lfloor\frac{t}{\tau }\rfloor \tau }\Theta_{\lfloor \frac{t}{\tau }\rfloor \tau }
                     \prod_{i=1}^{\lfloor \frac{t}{\tau }\rfloor} \Theta_{(i-1)\tau }^*B_{\tau }^*\mathcal{U}_{\tau } \Theta_{(i-1)\tau }.\\
\end{split}\end{equation}

Now let $\varphi,\varphi_k \in \mathcal{H}\otimes\mathcal{F}$ be such that 
$\varphi_k \xrightarrow{k\to\infty} \varphi$. Note that then
\begin{equation}\begin{split}\label{eq convergence mathcalRtkphik}
\|\mathcal{R}_t^{k}\varphi_k -\mathcal{U}_t\varphi\|{ } &= 
                                         \|\mathcal{R}_t^k\varphi_k  -\mathcal{U}_t\varphi  + \mathcal{R}_t^k\varphi -\mathcal{R}_t^k\varphi\| \\
                                    &\le \|\mathcal{R}_t^k(\varphi_k-\varphi)\| + \|(\mathcal{R}_t^k-\mathcal{U}_t)\varphi\| 
                                            =  \|\varphi_k-\varphi\| + \|(\mathcal{R}_t^k-\mathcal{U}_t)\varphi\| \xrightarrow{k\to\infty} 0,
\end{split}\end{equation}
where in the last step we have used condition 2.

Let $\psi \in \mathcal{H}\otimes\mathcal{F}$. If we define $\varphi_k = \Theta_{\tau  } B_{\tau }^*\mathcal{R}_{\tau }^k\psi$ and 
$\varphi = \Theta_{\tau } B_{\tau }^*\mathcal{U}_{\tau }\psi$, then 
due to the unitarity of $\Theta_{\tau }$ and $B_{\tau }^*$ and due to the assumed condition 2, 
we find $\varphi_k  \xrightarrow{k\to\infty} \varphi$. Equation \eqref{eq convergence mathcalRtkphik} then yields 
that $\mathcal{R}_{\tau}^k\varphi_k \xrightarrow{k\to\infty} \mathcal{U}_{\tau }\phi$. We can now redefine 
$\varphi_k =  \Theta_{\tau } B_{\tau }^*\mathcal{R}_{\tau }^k\Theta_{\tau  }B_{\tau }^*\mathcal{R}_{\tau }^k\psi$ 
and $\varphi = \Theta_{\tau } B_{\tau }^* \mathcal{U}_{\tau }\Theta_{\tau } B_{\tau }^* \mathcal{U}_{\tau } \psi$
and we immediately have $\varphi_k \xrightarrow{k\to\infty} \varphi$. Using equation \eqref{eq convergence mathcalRtkphik} as before, 
we can then conclude $\mathcal{R}_{\tau }^k\varphi_k \xrightarrow{k\to\infty} \mathcal{U}_{\tau } \varphi$. By iterating this 
procedure sufficiently many times, we find using the expressions in equation  \eqref{eq expressions Rtk U} 
\begin{equation*}
	\lim_{k\to\infty}
	\|R_t^{k}\psi- U_t\psi\|=0,\qquad \forall \psi \in\mathcal{H}\otimes\mathcal{F}.
\end{equation*}
The implication $2\Rightarrow 1$ now follows from the implication $2 \Rightarrow 1$ of Theorem \ref{theorem BoutenvanHandel}.

$2 \Rightarrow 3:$
We define the following 
unitary operator on $\mathcal{H}\otimes \mathcal{F}$
\begin{equation*}
\hat{\mathcal{R}}^k = \Theta_{\tfrac{\tau}{2^k}}B_{\tfrac{\tau}{2^k}}R^k, 
\end{equation*}
where $R^k$ is the unitary of Definition \ref{definition discrete evolution}. Note 
that 
\begin{equation*}
(\hat{\mathcal{R}}^k)^{\left\lfloor \tfrac{t 2^k}{\tau }\right\rfloor} = 
\Theta_{\left\lfloor \tfrac{t 2^k}{\tau }\right\rfloor \tfrac{\tau }{2^k}}\mathcal{R}_{\left\lfloor \tfrac{t 2^k}{\tau }\right\rfloor \tfrac{\tau }{2^k}}^k.
\end{equation*}
This means that for any $t$ of the form $t= r\tau $ with a 
dyadic rational $r = \tfrac{p}{2^q}$ with $p,q\in\mathbb{N}$, we have 
that for sufficiently large $k$
\begin{equation*}
\|\mathcal{R}_t^k\psi - \mathcal{U}_t\psi\| = \|\Theta_t \mathcal{R}_t^k\psi - \Theta_t\mathcal{U}_t\psi\| = \|(\hat{\mathcal{R}}^k)^{\lfloor r2^k\rfloor}\psi - \hat{\mathcal{U}}_t\psi\|, \qquad \forall \psi \in \mathcal{H}\otimes\mathcal{F}
\end{equation*}
since $\Theta_t$ is unitary and $2^{k-q} \in \mathbb{N}$ for sufficiently large $k$. 
Using the strong convergence of the assumed condition 2, we find 
that for every $t= r\,\tau$ with $r$ a dyadic rational, we have
\begin{equation}\label{eq hatmathcalRk}
\left\|(\hat{\mathcal{R}}^k)^{\left\lfloor t\tfrac{2^k}{\tau }\right\rfloor}\psi - \hat{\mathcal{U}}_t\psi\right\|\xrightarrow{k\to\infty} 0,\qquad \forall \psi \in \mathcal{H}\otimes \mathcal{F}.
\end{equation}
Since $\hat{U}_t$ is strongly continuous in $t$ and the dyadic rationals are dense in $\mathbb{R}$, 
we see that equation  \eqref{eq hatmathcalRk} holds for every fixed $t \ge 0$. We can now apply the implication 
$2 \Rightarrow 3$ of Theorem \ref{theorem Kurtz}.

$3 \Rightarrow 2:$ Trivial.
\qed

{\bf Acknowledgement.}
LB thanks Bas Janssens for stimulating discussion and a critical reading of an early version of this article.

\bibliography{discrete_SLH}
\end{document}